\documentclass{article}

\usepackage[utf8]{inputenc}
\usepackage{fullpage}
\usepackage{url}
\usepackage{hyperref}
\usepackage[ruled,vlined,linesnumbered]{algorithm2e}
\usepackage{amsthm}
\usepackage{graphicx}
\usepackage{subfig}
\usepackage{booktabs}
\usepackage{CJKutf8}
\usepackage{amsmath}
\usepackage{multirow}
\usepackage{newtxtext,newtxmath}
\usepackage[numbers,sort]{natbib}
\usepackage{authblk}
\usepackage{bbding}

\usepackage{xcolor}
\hypersetup{
    colorlinks=true,
    citecolor=blue,
    linkcolor=red,
    urlcolor=orange,
}
\newcommand{\Str}[1]{\texttt{#1}}
\newcommand{\Name}[1]{\textsf{#1}}
\newcommand{\Times}[1]{{#1}$\times$}
\newcommand{\Ceil}[1]{\lceil{#1}\rceil}
\newcommand{\Floor}[1]{\lfloor{#1}\rfloor}

\theoremstyle{plain}

\newtheorem{theorem}{Theorem}

\SetCommentSty{mycommfont}
\SetArgSty{\rmdefault}
\SetFuncSty{\rmdefault}
\SetKwInOut{Input}{Input}
\SetKwInOut{Output}{Output}

\newcommand\blfootnote[1]{%
  \begingroup
  \renewcommand\thefootnote{}\footnote{#1}%
  \addtocounter{footnote}{-1}%
  \endgroup
}

\title{Engineering faster double-array Aho--Corasick automata$^*$}

\author[1]{Shunsuke Kanda}
\author[1]{Koichi Akabe}
\author[2]{Yusuke Oda}
\affil[1]{LegalOn Technologies Research, Japan}
\affil[2]{Center for Data-driven Science and Artificial Intelligence, Tohoku University, Japan}

\date{}

\begin{document}

\maketitle

\blfootnote{$^*$This is the peer reviewed version of the following article: ``Shunsuke Kanda, Koichi Akabe, and Yusuke Oda. Engineering faster double-array Aho--Corasick automata. Software: Practice and Experience (SPE), 53(6): 1332--1361, 2023,'' which has been published in final form at \url{https://doi.org/10.1002/spe.3190}. This article may be used for non-commercial purposes in accordance with Wiley Terms and Conditions for Use of Self-Archived Versions. This article may not be enhanced, enriched or otherwise transformed into a derivative work, without express permission from Wiley or by statutory rights under applicable legislation. Copyright notices must not be removed, obscured or modified. The article must be linked to Wiley’s version of record on Wiley Online Library and any embedding, framing or otherwise making available the article or pages thereof by third parties from platforms, services and websites other than Wiley Online Library must be prohibited.}

\begin{abstract}
\noindent
Multiple pattern matching in strings is a fundamental problem in text processing applications such as regular expressions or tokenization. This paper studies efficient implementations of \emph{double-array Aho--Corasick automata} (DAACs), data structures for quickly performing the multiple pattern matching. The practical performance of DAACs is improved by carefully designing the data structure, and many implementation techniques have been proposed thus far. A problem in DAACs is that comprehensive descriptions and experimental analyses on their ideas are not provided. Engineers face difficulties in implementing an efficient DAAC.

In this paper, we review implementation techniques for DAACs and provide a comprehensive description of them. We also propose several new techniques for further improvement. We conduct exhaustive experiments through real-world datasets and reveal the best combination of techniques to achieve a higher performance in DAACs. The best combination is different from those used in the most popular libraries of DAACs, which demonstrates that their performance can be further enhanced. On the basis of our experimental analysis, we developed a new Rust library for fast multiple pattern matching using DAACs, named \emph{Daachorse}, as open-source software at \url{https://github.com/daac-tools/daachorse}. Experiments demonstrate that Daachorse outperforms other AC-automaton implementations, indicating its suitability as a fast alternative for multiple pattern matching in many applications.
\end{abstract}

\section{Introduction}

Multiple pattern matching in strings is a fundamental problem in text and natural language processing \cite{navarro2002flexible}.
Given a set of patterns and a text, the goal of this problem is to report all occurrences of patterns in the text.
A representative application is regular expressions, which provide a powerful way to express richer patterns using meta characters, including repetitions and wildcards.
Another is tokenization in unsegmented natural languages such as Japanese and Chinese, which partitions a sentence into shorter units called tokens or words.
Multiple pattern matching is essential in these applications, and its time efficiency is crucial.

The \emph{Aho--Corasick (AC) algorithm} \cite{aho1975efficient} is a fast solution for multiple pattern matching.
It uses an \emph{AC automaton} and performs the matching with $O(n)$ character comparisons, where $n$ is the length of an input text.
Despite the theoretical guarantee, the practical performance of the AC algorithm is significantly affected by its internal data structure used to represent the AC automaton \cite{nieminen2007efficient}.
Thus, carefully designing the internal data structure is vital to achieve faster matching.

\emph{Double-array AC automata (DAACs)} are AC-automaton representations for fast matching.
The core component is the \emph{double-array} \cite{aoe1989efficient}, a data structure to implement transition lookups in an optimal time.
Prior experiments \cite{nieminen2007efficient} demonstrated that DAACs were faster than various other representations.
Thanks to their time efficiency, double-array structures are used in a wide range of applications such as dictionary lookups \cite{darts,dartsclone}, compressed string dictionaries \cite{kanda2017compressed}, tokenization \cite{kudo-etal-2004-applying,song2021fast}, language models \cite{norimatsu2016fast,yasuhara2013efficient}, classification \cite{yoshinaga2014self}, and search engines \cite{groonga}.

To achieve a higher performance in DAACs, it is essential to carefully design the data structure with regard to the target applications and data characteristics.
In the three decades since the original idea of the double-array was proposed by Aoe \cite{aoe1989efficient}, many implementation techniques have been developed from different points of view, such as scalability \cite{fuketa2004new,kanda2017compressed,kanda2016compression}, cache efficiency \cite{yata2007compact,kanda2017rearrangement}, construction speed \cite{morita2001fast,oono2003fast,yata2008fast}, and specializations \cite{yata2008fast,norimatsu2016fast,yasuhara2013efficient}.
In addition to these academic studies, many open-source libraries have been developed, e.g., \cite{darts,dartsclone,Vonng-ac,AhoCorasickDoubleArrayTrie}, which sometimes contain original techniques.
The problem is that comprehensive descriptions and experimental analyses on their ideas are not provided, which makes it difficult for engineers to implement an efficient DAAC.

\paragraph{Our contributions}

In this paper, we review various implementation techniques in DAACs and provide a comprehensive description of them, including categorization for facilitating comparison (as summarized in Table \ref{tab:summary}).
We also propose several new techniques for further improvement.
We provide exhaustive experimental analysis through real-world datasets and reveal the best combination of implementation techniques.
The best combination is different from those used in the most popular libraries of DAACs \cite{AhoCorasickDoubleArrayTrie,Vonng-ac}, demonstrating that their performance can be further enhanced.

On the basis of our experimental analysis, we develop a new Rust library for fast multiple pattern matching using DAACs, named \emph{Daachorse}, as open-source software at \url{https://github.com/daac-tools/daachorse} under the Apache-2.0 or MIT license.\footnote{We also release a Python wrapper of Daachorse at \url{https://github.com/daac-tools/python-daachorse}.}
Our experiments demonstrate that Daachorse outperforms other AC-automaton implementations, indicating its suitability as a fast alternative for multiple pattern matching in many applications.
Daachorse has been plugged into \emph{Vaporetto} \cite{vaporetto}, a Japanese tokenizer written in Rust.
Vaporetto is a fast implementation of the pointwise prediction method \cite{neubig2011pointwise,shinnou2000deterministic,sassano2002empirical} that determines token boundaries using a discriminative model.
The core step of this method is feature extraction performed with the AC algorithm.
Its running time significantly affects the entire processing time, and the fast pattern matching provided by Daachorse assures the time efficiency of Vaporetto.
For example, Daachorse performs tokenization \Times{2.6} faster than other implementations, as demonstrated in this paper.

Our experiments were conducted with the system programming language Rust, although we do not believe that the results will impact our general conclusion.
This is because most computations in our algorithms are simple array manipulations, and the output binaries will not differ significantly from those using other system languages.\footnote{To generate succinct binaries, we carefully implemented core components (e.g., transition functions) by introducing optimization keywords and types (e.g., \texttt{inline}, \texttt{unsafe}, and \texttt{NonZeroU32}) and eliminated redundant instructions to measure the time performance.}

\section{Preliminaries}
\label{sect:pre}

\subsection{Basic definition and notation}
\label{sect:pre:basic}

A \emph{string} is a finite sequence of characters over a finite integer alphabet $\Sigma = \{ 0,1,\dots,|\Sigma| -1 \}$.
Our strings always start at position zero.
The empty string $\varepsilon$ is a string of length zero.
Given a string $P$ of length $n \geq 1$, $P[i..j)$ denotes the \emph{substring} $P[i],P[i+1],\dots,P[j-1]$ for $0 \leq i \leq j \leq n$.
Specially, $P[0..i)$ is a \emph{prefix} of $P$, and $P[i..n)$ is a \emph{suffix} of $P$ for $0 \leq i \leq n$.
Let $|P| := n$ denote the length of $P$.
The same notation is applied to \emph{arrays}.
We denote the cardinality of a set $A$ by $|A|$.

\subsection{Multiple pattern matching}
\label{sect:pre:mpm}

Given a set of strings $\mathcal{D} = \{ P_1, P_2, \dots, P_{|\mathcal{D}|} \}$ and a string $T$, the goal of multiple pattern matching is to report all occurrences $\{ (k,i,j) : P_k \in \mathcal{D}, P_k = T[i..j) \}$, where an occurrence $(k,i,j)$ consists of the index $k$ of the matched pattern $P_k$ and the starting and ending positions $i,j$ appearing in $T$.
Throughout this paper, we refer to $\mathcal{D}$ as a \emph{dictionary}, $P_k$ as a \emph{pattern}, and $T$ as a \emph{text}.

Table \ref{tab:dictioanry} shows an example of a dictionary $\mathcal{D}$ consisting of six patterns.
In all examples throughout this paper, we denote indices of patterns by upper-case letters instead of numbers and integer elements in $\Sigma$ by lower-case letters.
Given the dictionary and text $T[0..6) = \Str{abacdd}$, the occurrences are $(\Str{A},0,2)$, $(\Str{B},1,2)$, $(\Str{D},1,4)$, and $(\Str{F},4,6)$.

\begin{table}[tb]
\footnotesize
\centering
\caption{
A dictionary $\mathcal{D}$ of six patterns (used in examples throughout this paper).
Patterns are indexed with upper-case letters $\Str{A},\Str{B},\Str{C},\dots$ instead of numbers in the examples.
Elements in the integer alphabet $\Sigma$ are denoted with lower-case letters.
$\Sigma = \{ \Str{a} = 0,\Str{b} = 1,\Str{c} = 2,\Str{d} = 3 \}$.
}
\label{tab:dictioanry}
\begin{tabular}{|c||c|}
\hline
Index & Pattern \\
\hline\hline
\Str{A} & \Str{ab} \\
\Str{B} & \Str{b} \\
\Str{C} & \Str{bab} \\
\Str{D} & \Str{bac} \\
\Str{E} & \Str{db} \\
\Str{F} & \Str{dd} \\
\hline
\end{tabular}
\end{table}

\subsection{Aho--Corasick algorithm}
\label{sect:pre:ac}

The AC automaton \cite{aho1975efficient} is a finite state machine to find all occurrences of patterns in a single scan of a text.
The AC automaton for a dictionary $\mathcal{D}$ is defined as the 5-tuple $(S,\Sigma,\delta,f,h)$:
\begin{itemize}
\item $S = \{ 0,1,\dots,|S|-1 \}$ is a finite set of states, where each state is identified by an integer, and the initial state is indicated by 0;
\item $\Sigma = \{ 0,1,\dots,|\Sigma| -1 \}$ is the alphabet;
\item $\delta: S \times \Sigma \rightarrow S \cup \{ -1 \}$ is a transition function, where $-1$ is an invalid state id;
\item $f: S \setminus \{ 0 \} \rightarrow S$ is a failure function; and
\item $h: S \rightarrow \mathcal{P}(\{ 1,2,\dots, | \mathcal{D} | \})$ is an output function, where $\mathcal{P}(\cdot)$ is the power set.
\end{itemize}
Figure \ref{fig:ac:pma} shows an example of the AC automaton.

\begin{figure}[tb]
\centering
\subfloat[AC automaton]{
\includegraphics[scale=0.8]{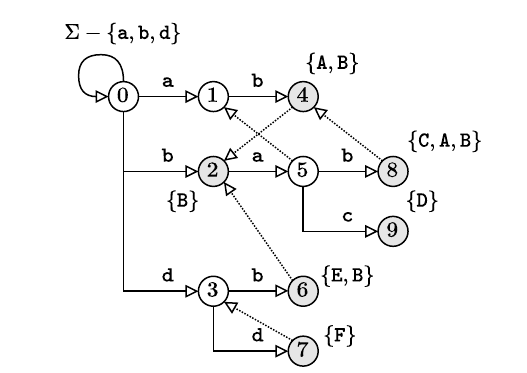}
\label{fig:ac:pma}
}
\subfloat[Trie]{
\includegraphics[scale=0.8]{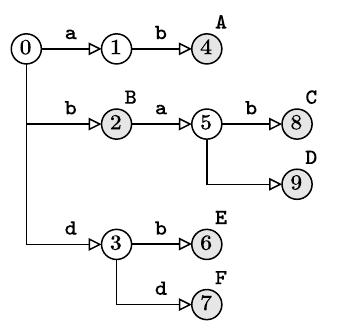}
\label{fig:ac:trie}
}
\caption{
Examples of (a) an AC automaton for the dictionary of Table \ref{tab:dictioanry} and (b) its trie part.
Transitions are depicted by solid line arrows.
$\delta(0,\Str{b}) = 2$, $\delta(2,\Str{a}) = 5$, and $\delta(2,\Str{c}) = -1$.
We depict the mappings of the failure function (except ones to the initial state) by dotted line arrows.
$f(4)=2$, $f(5)=1$, $f(6)=2$, $f(7)=3$, $f(8)=4$, and $f(s) = 0$ for the other states $s$.
Output states are shaded and associated with pattern indices (drawn from $\Str{A},\Str{B},\Str{C},\dots$).
$h(2)=\{\Str{B}\}$, $h(4)=\{\Str{A,B}\}$, $h(6)=\{\Str{E,B}\}$, $h(7)=\{\Str{F}\}$, $h(8)=\{\Str{C,A,B}\}$, $h(9)=\{\Str{D}\}$, and $h(s) = \emptyset$ for the other states $s$.
}
\label{fig:ac}
\end{figure}

The transition function $\delta$ is built on a \emph{trie} for the dictionary $\mathcal{D}$.
The trie \cite{fredkin1960trie} is a tree automaton formed by merging the prefixes of patterns in $\mathcal{D}$.
Figure \ref{fig:ac:trie} shows the trie in the AC automaton.
A state $s$ in the trie represents any prefix of patterns in $\mathcal{D}$, and the prefix can be extracted by concatenating transition labels from the initial state to state $s$.
We denote by $\phi(s)$ the string represented by state $s$.
For example, $\phi(4) = \Str{ab}$ and $\phi(5) = \Str{ba}$ in Figure \ref{fig:ac:trie}.
The initial state always represents the empty string $\varepsilon$, i.e., $\phi(0) = \varepsilon$.
We call states satisfying $\phi(s) \in \mathcal{D}$ \emph{output states}.
If a state $s$ does not indicate any other state with character $c$, $
\delta(s,c) = -1$ is defined using the invalid state id $-1$.
There is one difference in the definition of $\delta$ between the AC automaton and trie:
the AC automaton redefines special transitions $\delta(0,c) := 0$ for labels $c$ such that $\delta(0,c) = -1$ in the trie.

The failure function $f$ maps a state $s \in (S \setminus \{0\})$ to another state $t \in (S \setminus \{s\})$ such that $\phi(t)$ is a longer suffix of $\phi(s)$ than $\phi(t')$ for $t' \in (S \setminus \{s,t\})$.
$f(s)$ is used for characters $c \in \Sigma$ such that $\delta(s,c) = -1$.
Note that, since the empty string $\varepsilon$ is a suffix of any string, the initial state is always one of the candidates.
Let $F(s)$ be a set of output states reached through only the failure function from state $s$.
For example, $F(8) = \{8,4,2\}$ in Figure \ref{fig:ac:pma}.
The output function $h(s)$ is the set of pattern indices associated with output states in $F(s)$,
i.e., $h(s) = \{ k : P_k \in \mathcal{D}, P_k =  \phi(t), t \in F(s)\}$.

\paragraph{Matching algorithm}

\begin{algorithm}[tb]
\small
\DontPrintSemicolon
\Input{Text $T$ of length $n$}
\Output{All occurrences of patterns $\{ (k,i,j) : P_k \in \mathcal{D}, P_k = T[i..j) \}$}
$s \gets 0$\;
\For{$j = 0,1,\dots,n-1$}{
    \For{$k \in h(s)$}{
        $i \gets j - |P_k|$\;
        Output an occurrence $(k,i,j)$\;
    }
    $s \gets \delta^*(s,T[j])$\;
}
\For{$k \in h(s)$}{
    $i \gets n - |P_k|$\;
    Output an occurrence $(k,i,n)$\;
}
\caption{Multiple pattern matching using AC automaton.}
\label{algo:ac}
\end{algorithm}

Algorithm \ref{algo:ac} shows the AC algorithm that performs multiple pattern matching using the AC automaton.
The algorithm uses the extended transition function $\delta^*$:
\begin{equation}
\delta^*(s,c)= \begin{cases}
\delta(s,c) & \textrm{if~} \delta(s,c) \neq -1 \\
\delta^*(f(s),c) & \textrm{otherwise}.
\end{cases}
\end{equation}

Given a text $T$, the AC algorithm scans $T$ character by character, visits states from the initial state $s = 0$ with $\delta^*$, and reports occurrences where patterns in $h(s)$ are associated with each visited state $s$.
For example, assume we are scanning text $T[0..6)=\Str{abacdd}$ with the AC automaton in Figure \ref{fig:ac:pma}.
We first move with \Str{ab} from state 0 to state 4 by $\delta^*(0,\Str{a}) = 1$ and $\delta^*(1,\Str{b}) = 4$, and then report two occurrences, $(\Str{A},0,2)$ and $(\Str{B},1,2)$, associated with $h(4)$.
Next, we move to state 9 by $\delta^*(4,\Str{a}) = 5$ and $\delta^*(5,\Str{c}) = 9$ and report occurrence $(\Str{D},1,4)$ associated with $h(9)$.
Last, we move to state 7 with $\delta^*(9,\Str{d}) = 3$ and $\delta^*(3,\Str{d}) = 7$ and report occurrence $(\Str{F},4,6)$ associated with $h(7)$.

In the scanning, the number of states visited with $\delta$ and $f$ is bounded by $2n$ for a text of length $n$.
The algorithm runs in $O(n+\textsf{occ})$ time in the most efficient case, where $\textsf{occ}$ is the number of occurrences.

\subsection{Double-array Aho--Corasick automata (DAACs)}
\label{sect:pre:da}

\begin{figure}[tb]
\centering
\includegraphics[scale=0.85]{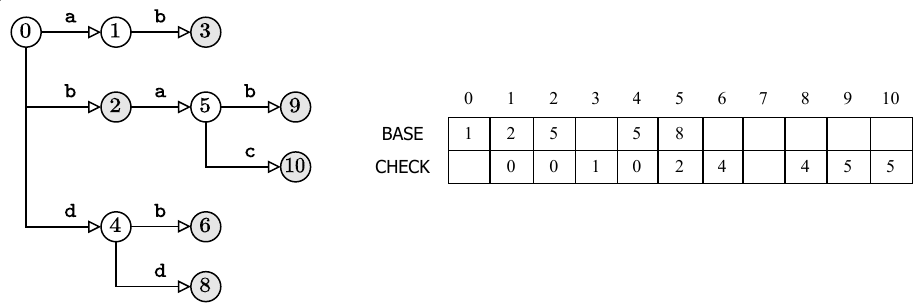}
\caption{
\Name{BASE} and \Name{CHECK} implementing the transition function $\delta$ of Figure \ref{fig:ac:trie}.
$\Sigma = \{ \Str{a} = 0,\Str{b} = 1,\Str{c} = 2,\Str{d} = 3 \}$.
The state ids are assigned to satisfy Equation \eqref{eq:da}.
$\delta(0,\Str{b}) = 2$ is simulated by $\Name{BASE}[0] + \Str{b} = 2$ and $\Name{CHECK}[2] = 0$.
$\delta(2,\Str{d}) = -1$ is simulated by $\Name{BASE}[2] + \Str{d} = 8$ and $\Name{CHECK}[8] \neq 2$.
The state id 7 is a vacant id because its element does not represent any state of the original trie.
}
\label{fig:da:trie}
\end{figure}

The \emph{double-array} \cite{aoe1989efficient} is a data structure to implement the transition function $\delta$ using two one-dimensional arrays: \emph{\Name{BASE}} and \emph{\Name{CHECK}} (see Figure \ref{fig:da:trie}).
The double-array arranges original states in $S$ onto \Name{BASE} and \Name{CHECK} and assigns new state ids to the original states.
An element of \Name{BASE} and \Name{CHECK} corresponds to an original state, and its array offsets indicate new state ids.
The \Name{BASE} and \Name{CHECK} arrays are constructed so that new state ids satisfy the following equations when $\delta(s,c) = t$ (except for $t = -1$):
\begin{equation}
\label{eq:da}
\Name{BASE}[s] + c = t ~ \textrm{and} ~ \Name{CHECK}[t] = s.
\end{equation}

Equation \eqref{eq:da} enables us to look up a transition $\delta(s,c)$ in two very simple steps:
(i) computing $t := \Name{BASE}[s] + c$ and (ii) returning $t$ if $\Name{CHECK}[t] = c$ or $-1$ otherwise.
The time complexity is $O(1)$.

Let $S_{\Name{BC}} = \{ 0,1,\dots,|S_{\Name{BC}}|-1 \}$ be a set of state ids in a resulting double-array structure.
The space complexity of \Name{BASE} and \Name{CHECK} is $O(|S_{\Name{BC}}|)$, where $|S_{\Name{BC}}| = |S|$ in the best case and $|S_{\Name{BC}}| = |S||\Sigma|$ in the worst case, as $S_{\Name{BC}}$ can contain unused ids to satisfy Equation \eqref{eq:da}.
For example, state id 7 is unused in Figure \ref{fig:da:trie}.
We call such state ids \emph{vacant ids}.
Constructing \Name{BASE} and \Name{CHECK} as few vacant ids as possible is important for space efficiency.

\paragraph{Construction of \Name{BASE} and \Name{CHECK}}

Constructing \Name{BASE} and \Name{CHECK} with an original trie implementing $\delta$ requires visiting states from the root and searching for values of $\Name{BASE}[s]$ for each state $s$ to satisfy Equation \eqref{eq:da}.
More formally, for a state $s$ having outgoing transitions with $k$ labels $c_1,c_2,\dots, c_k $, we search for a \Name{BASE} value $b$ such that $b + c_j$ is a vacant id for each $1 \leq j \leq k$ and define $\Name{BASE}[s] = b$ and $\Name{CHECK}[b + c_j] = s$, which we call a \emph{vacant search}.
In Section \ref{sect:tech:constr}, we will discuss approaches to accelerate vacant searches.

\paragraph{Extension to AC automaton}

We can extend the double-array to DAACs by introducing components for failure and output functions.
Figure \ref{fig:da} shows an example DAAC for the AC automaton in Figure \ref{fig:ac:pma}.
The failure function $f$ is implemented with array \Name{FAIL}{} such that $\Name{FAIL}[s] = f(s)$.
Using the arrays \Name{BASE}, \Name{CHECK}, and \Name{FAIL}, the extended transition function $\delta^{*}$ is implemented as Algorithm \ref{algo:daac-delta}.

\begin{figure}[tb]
\centering
\includegraphics[scale=0.85]{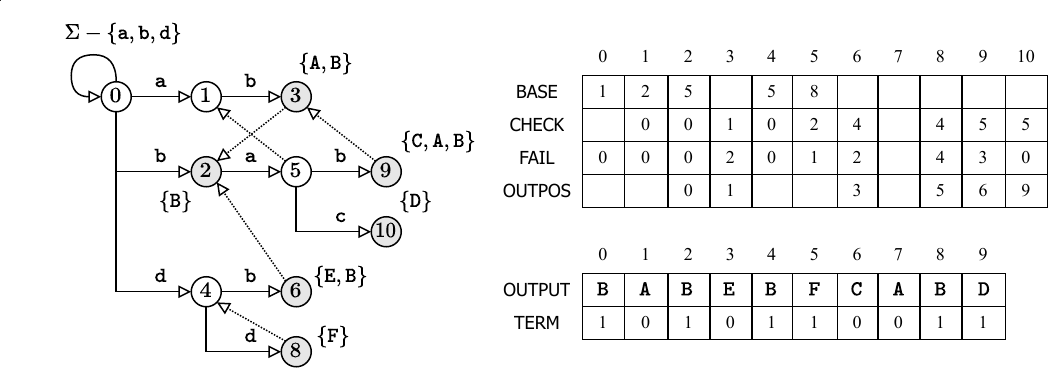}
\caption{
DAAC for the AC automaton in Figure \ref{fig:ac:pma}.
}
\label{fig:da}
\end{figure}

\begin{algorithm}[tb]
\small
\DontPrintSemicolon
\Repeat{}{
    $t \gets \Name{BASE}[s] + c$\;
    \uIf{$\Name{CHECK}[t] = s$}{
        \KwRet $t$\;
    }
    \ElseIf{$s = 0$}{
        \KwRet $0$\;
    }
    $s \gets \Name{FAIL}[s]$\;
}
\caption{Implementation of $\delta^{*}(s,c)$ using DAAC.}
\label{algo:daac-delta}
\end{algorithm}

The output function $h$ is implemented with three arrays:
\Name{OUTPUT} arranges values in $h(s)$ for output states $s$,
\Name{TERM} stores bit flags to identify terminals of each set $h(s)$ in \Name{OUTPUT}, and
$\Name{OUTPOS}[s]$ stores the starting positions of each set $h(s)$ in \Name{OUTPUT}.
In construction, we visit each output state $s$, append values in $h(s)$ to \Name{OUTPUT} and $|h(s)| - 1$ 0s and one 1 to \Name{TERM}, and set the starting positions in $\Name{OUTPOS}[s]$ accordingly.
Using these arrays, we can extract $h(s)$ by scanning $\Name{OUTPUT}[i..j]$ from the starting position $i = \Name{OUTPOS}[s]$ until encountering $\Name{TERM}[j] = 1$.

\begin{algorithm}[tbp]
\small
\DontPrintSemicolon
\Input{Dictionary $\mathcal{D} = \{ P_1, P_2, \dots, P_{|\mathcal{D}|} \}$}
\Output{\Name{BASE}, \Name{CHECK}, \Name{FAIL}, \Name{OUTPOS}, \Name{OUTPUT}, and \Name{TERM}, accepting $\mathcal{D}$}
Construct an original AC automaton $(S,\Sigma,\delta,f,h)$ from $\mathcal{D}$ using a simple data structure such as the linked list form \cite{nieminen2007efficient}, with Algorithms 2 and 3 in \cite{nieminen2007efficient}\;
Initialize \Name{BASE}, \Name{CHECK}, \Name{FAIL}, and \Name{OUTPOS} as arrays with enough elements\;
Initialize \Name{OUTPUT} and \Name{TERM} as empty arrays\;
Initialize $\psi$ as an array with $|S|$ elements\tcc*[r]{Mapping from state ids in $S$ to new state ids of the DAAC}
$Q \gets \{ (0,0) \}$\tcc*[r]{Queue to traverse the original automaton, initialized with the initial state ids}
\While{$Q \neq \emptyset$}{
    Pop $(s,s')$ from $Q$\tcc*[r]{$s \in S$ and $s'$ is the corresponding state id of the DAAC}
    $\psi[s] \gets s'$\;
    \If{$s \neq 0$}{
        $\Name{FAIL}[s'] \gets \psi[f(s)]$\tcc*[r]{$\psi[f(s)]$ is always valid because of the breadth-first traversal using $Q$}
    }
    \If{$h(s) \neq \emptyset$}{
        \eIf{$s \neq 0 \land h(s) = h(f(s))$}{
            $\Name{OUTPOS}[s'] \gets \Name{OUTPOS}[\Name{FAIL}[s']]$\;
        }{
            $\Name{OUTPOS}[s'] \gets |\Name{OUTPUT}|$\;
            Append values in $h(s)$ to \Name{OUTPUT}\;
            Append $|h(s)|-1$ 0s and one 1 to \Name{TERM}\;
        }
    }
    $E \gets$ a set of labels of outgoing transitions from state $s$\;
    \If{$E \neq \emptyset$}{
        $\Name{BASE}[s'] \gets$ integer $b$ such that $b + c$ is a vacant id for each $c \in E$\tcc*[r]{Vacant search}
        \For{$c \in E$}{\nllabel{algo:constr:line:edges}
            $(t,t') \gets (\delta(s,c), \Name{BASE}[s'] + c)$\;
            $\Name{CHECK}[t'] \gets s'$\;
            Push $(t, t')$ to $Q$\;
        }
    }
}
Output \Name{BASE}, \Name{CHECK}, \Name{FAIL}, \Name{OUTPOS}, \Name{OUTPUT}, and \Name{TERM}\;
\caption{Construction of the data structure of DAAC.}
\label{algo:constr}
\end{algorithm}

Algorithm \ref{algo:constr} shows the construction algorithm of DAAC.
The algorithm first construct an original AC automaton from an input dictionary using a simple data structure and then converts it into the DAAC representation.

\section{Implementation techniques in DAACs}
\label{sect:tech}

In this section, we review implementation techniques to improve the DAAC performance and describe them based on our categorization (summarized in Table \ref{tab:summary}). 

\subsection{Management of output sets}
\label{sect:tech:out}

We first describe how to manage output sets efficiently.
Figure \ref{fig:da} shows a simple approach that arranges values in $h(s)$ in an array.
We call this approach \Name{Simple}.
An advantage of \Name{Simple} is the \emph{locality of reference} that can extract $h(s)$ by a sequential scan.
However, \Name{Simple} can maintain many duplicate values in \Name{OUTPUT}.
For example, value \Str{B} appears four times in the \Name{OUTPUT} of Figure \ref{fig:da}.
The length of \Name{OUTPUT} is bounded by $O(|\mathcal{D}| \cdot K)$, where $K$ is the average length of patterns in the dictionary.
In the following, we present two additional approaches, \Name{Shared} and \Name{Forest}, to improve the memory efficiency of \Name{Simple}.

\begin{figure}[tb]
\centering
\subfloat[\Name{Shared}]{
\includegraphics[scale=0.9]{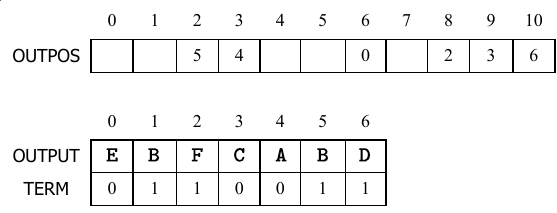}
\label{fig:out:shared}
}\\
\subfloat[\Name{Forest}]{
\includegraphics[scale=0.9]{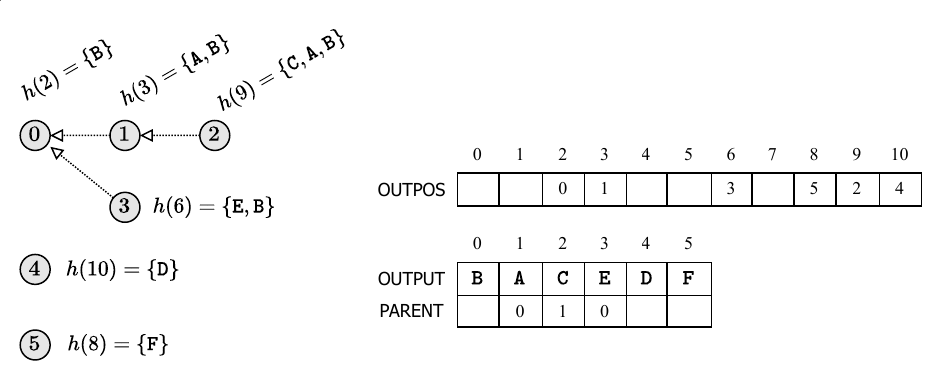}
\label{fig:out:forest}
}
\caption{
Examples of approaches to store output sets of Figure \ref{fig:da}.
}
\label{fig:out}
\end{figure}

\paragraph{Shared arrangement}

To design \Name{Shared}, we exploit the fact that
$h(t) \subseteq h(s)$ when state $t$ can be reached from state $s$ through only failure functions.
For example, $h(2) \subseteq h(3) \subseteq h(9)$ in Figure \ref{fig:da}.
This fact indicates that values in $h(t)$ can be represented as a part of the values in $h(s)$.
\Name{Shared} implements the \Name{OUTPUT} and \Name{TERM} arrays while sharing such common parts.
Figure \ref{fig:out:shared} shows an example of \Name{Shared}, where $h(2)$ and $h(3)$ are represented in the rear parts of $\Name{OUTPUT}[3..5]$ for $h(9)$.
\Name{Shared} can extract $h(s)$ in the same manner as \Name{Simple} and does not lose the locality of reference.

In construction, we find shareable common parts using a greedy algorithm.
We visit output states $s$ (such that $h(s) \neq \emptyset$) in order of depth, starting with the deepest state.
For each visited state $s$, we share the output sets $h(t)$ with $h(s)$, where states $t$ are output states reachable through only failure functions from $s$;
however, if $h(t)$ has already been shared with another output set, we skip the sharing.

\paragraph{Forest representation}

A drawback of \Name{Shared} is that it cannot remove all duplicate values, and the space complexity of \Name{OUTPUT} and \Name{TERM} is not improved.
For example, value \Str{B} still appears twice in Figure \ref{fig:out:shared}.
\Name{Forest} is an alternative approach that maintains only unique pattern indices \cite{belazzougui2010succinct}.

Multiple trees, or a \emph{forest}, can be constructed from an AC automaton by chaining output states through the failure function, which we call an \emph{output forest}.
The left part of Figure \ref{fig:out:forest} shows such an output forest constructed from the AC automaton in Figure \ref{fig:da}.
In the output forest, we can extract values in $h(s)$ by climbing up the corresponding tree to the root.

\Name{Forest} constructs an output forest whose nodes are indexed by numbers from $\{0,1,\dots,|\mathcal{D}| - 1\}$.
Our data structure represents the forest using two arrays, \Name{OUTPUT} and \Name{PARENT}, such that $\Name{OUTPUT}$ stores pattern indices and $\Name{PARENT}$ stores parent positions.
$\Name{OUTPOS}[s]$ stores the node index corresponding to an output state $s$ in the AC automaton.
$h(s)$ is extracted by visiting $\Name{OUTPUT}[i]$ from node $i := \Name{OUTPOS}[s]$ while updating $i := \Name{PARENT}[i]$ until encountering the root.
The right part of Figure \ref{fig:out:forest} shows the data structure.
Algorithm \ref{algo:forest} shows the algorithm to extract $h(s)$ for a given state $s$.

\begin{algorithm}[tbp]
\small
\DontPrintSemicolon
\eIf{$\Name{OUTPOS}[s] = \Name{NIL}$}{
    Output an empty set\;
}{
    $i \gets \Name{OUTPOS}[s]$\;
    $H \gets \{ \Name{OUTPUT}[i] \}$\;
    \While{$\Name{PARENT}[i] \neq \Name{NIL}$}{
        $i \gets \Name{PARENT}[i]$\;
        $H \gets H \cup \{ \Name{OUTPUT}[i] \}$\;
    }
    Output $H$\;
}
\caption{Extracting $h(s)$ with \Name{Forest} data structures. \Name{NIL} in \Name{OUTPOS} and \Name{PARENT} means that it does not indicate any position.}
\label{algo:forest}
\end{algorithm}

The advantage of \Name{Forest} is that the space complexity of \Name{OUTPUT} and \Name{PARENT} is bounded by $O(|\mathcal{D}|)$.
However, traversing an output forest requires random accesses, which sacrifices the locality of reference in \Name{Simple} and \Name{Shared}.
Also, the pointer array \Name{PARENT} consumes a larger space than the bit array \Name{TERM}.

\subsection{Byte- and character-wise automata}
\label{sect:tech:wise}

When handling strings consisting of multibyte characters, there are two representations of strings:
\begin{itemize}
\item \Name{Bytewise} represents strings as sequences of bytes in the UTF-8 format and uses byte values for transition labels. The maximum value in the alphabet $\Sigma$ is the maximum byte value appearing in input patterns.
\item \Name{Charwise} represents strings as sequences of code points in Unicode and uses code-point values for transition labels. The maximum value in the alphabet $\Sigma$ is the maximum code-point value appearing in input patterns.
\end{itemize}

For example, the Japanese string ``\begin{CJK}{UTF8}{ipxm}世界\end{CJK}'' (meaning ``the world'') is represented as a sequence of six bytes ``0xE4, 0xB8, 0x96, 0xE7, 0x95, 0x8C'' in \Name{Bytewise} and a sequence of two code points ``U+4E16, U+754C'' in \Name{Charwise}.
The alphabet $\Sigma$ in \Name{Bytewise} is $\{\textrm{0x00},\textrm{0x01},\dots,\textrm{0xE7}\}$, and
the alphabet $\Sigma$ in \Name{Charwise} is $\{\textrm{U+0000},\textrm{U+0001},\dots,\textrm{U+754C}\}$.

In the following, we first describe the advantages and disadvantages when constructing DAACs from strings in the \Name{Bytewise} or \Name{Charwise} scheme;
then, we present an approach called \Name{Mapped} to overcome the disadvantages in the \Name{Charwise} scheme.

\paragraph{Advantages and disadvantages}

There is a trade-off between the alphabet size (or $|\Sigma|$) and the number of automaton states (or $|S|$):
the alphabet size in \Name{Bytewise} is bounded by $2^{8}$ and is smaller than that in \Name{Charwise} bounded by $2^{21}$ (since code points are drawn up to {U+10FFFF}), and
the number of states in \Name{Charwise} is lower than that in \Name{Bytewise}.
For example, consider an AC automaton for the single pattern ``\begin{CJK}{UTF8}{ipxm}世界\end{CJK}''.
The alphabet size in \Name{Bytewise} is 232 (i.e., 0xE7 plus 1) while that in \Name{Charwise} is 30029 (i.e., U+754C plus 1);
thus, the number of states in \Name{Bytewise} is seven while that in \Name{Charwise} is three.

Smaller alphabets enable DAACs to achieve faster construction because the possible number of outgoing transitions from a state is suppressed, which facilitates vacant searches.
Fewer states enable DAACs to achieve faster matching because of suppressing the number of random memory accesses during a matching.
Therefore, \Name{Bytewise} is beneficial for faster construction, and \Name{Charwise} is beneficial for faster matching.

The memory usage of DAACs depends on the number of resulting double-array states $|S_{\Name{BC}}|$, which is the sum of the number of original states $|S|$ and that of vacant ids.
Although \Name{Charwise} can construct an AC automaton with fewer states, its large alphabet can produce more vacant ids because of the difficulty in vacant searches.
Thus, we can achieve memory efficiency by constructing a DAAC with fewer vacant ids in the \Name{Charwise} scheme.

\paragraph{Code mapping}

The \Name{Mapped} approach can overcome the disadvantages in \Name{Charwise} by mapping code points to smaller integers \cite{liu2011compression}.
This approach assigns code values to characters $c \in \Sigma$ with a certain frequency to assign smaller code values to more frequent characters.

We give a formal description of \Name{Mapped}.
Let us denote by $\#(c)$ the number of occurrences of character $c \in \Sigma$ in patterns of $\mathcal{D}$.
We construct the \emph{mapping function} $\pi: \Sigma \rightarrow \Sigma_{\pi}$ such that:
\begin{itemize}
\item $\Sigma_{\pi} = \{ -1, 0, 1, \dots, \sigma - 1 \}$, where $\sigma$ is the number of characters $c$ such that $\#(c) \neq 0$;
\item $\pi(c) = -1$ for characters $c$ when $\#(c) = 0$; and
\item $\pi(c)$ is the number of other characters $c'$ such that $\#(c) < \#(c')$ (breaking ties arbitrarily).
\end{itemize}
We call $\sigma$ the \emph{mapped alphabet size}.

Also, Equation \eqref{eq:da} is modified into
\begin{equation}
\label{eq:mapda}
\Name{BASE}[s] + \pi(c) = t ~ \textrm{and} ~ \Name{CHECK}[t] = s.
\end{equation}
Note that characters $c$ such that $\pi(c)= -1$ are not used in transitions because (i) during construction, the characters do not appear, and (ii) during pattern matching, we can immediately know whether transitions with the characters fail.

If $\Sigma$ includes many characters $c$ such that $\#(c) = 0$, the mapped alphabet size $\sigma$ becomes much smaller than $|\Sigma|$.
Moreover, occurrences of real-world characters are often skewed following Zipf's law \cite{yihan-2021-meaningfulness}, indicating that most of the characters in $\mathcal{D}$ are represented with small code values via the mapping $\pi$.
Prior works have empirically demonstrated that mapping $\pi$ reduces the resultant vacant ids and shortens the construction time \cite{liu2011compression,norimatsu2016fast,yasuhara2013efficient}.

However, we need to store an additional data structure for implementing the mapping $\pi$.
We implement $\pi$ with an array of length $|\Sigma|$ such that the $c$-th element stores the value of $\pi(c)$.
When the array consists of $2^{21}$ four-byte integers, it takes 8 MiB of memory.
In experiments in Sections \ref{sect:exp} and \ref{sect:syst}, the $\pi$ data structure is implemented with a simple four-byte integer array, although Section \ref{sect:exp:wise} provides a comparison result between the cases of four- and three-byte arrays.

\subsection{Memory layout of arrays}
\label{sect:tech:layout}

This paper says that the data structure of DAACs consists of four arrays: \Name{BASE}, \Name{CHECK}, \Name{FAIL}, and \Name{OUTPOS}.
In the following, we describe the two memory layouts of the four arrays: \Name{Individual} and \Name{Packed}.

\begin{figure}[tb]
\centering
\subfloat[\Name{Individual}]{
\includegraphics[scale=0.9]{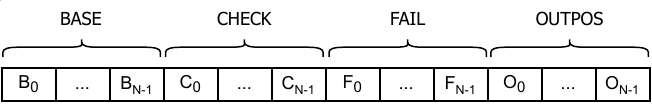}
\label{fig:layout:indi}
}\\
\subfloat[\Name{Packed}]{
\includegraphics[scale=0.9]{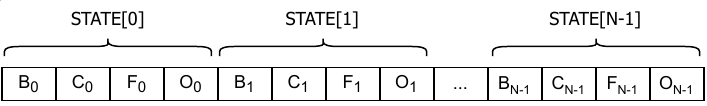}
\label{fig:layout:pack}
}
\caption{
Illustrations of memory layouts of arrays.
}
\label{fig:layout}
\end{figure}

\paragraph{Individual layout}

The \Name{Individual} layout maintains the four arrays individually, as shown in Figure \ref{fig:layout:indi}.

\paragraph{Packed layout}

Since the values of $\Name{BASE}[s]$, $\Name{CHECK}[s]$, $\Name{FAIL}[s]$, and $\Name{OUTPOS}[s]$ are accessed consecutively during a matching,
we should be able to improve the locality of reference by placing these values in a consecutive memory array.
The \Name{Packed} layout represents states using a one-dimensional array \Name{STATE}, where each element consists of the four fields \Name{base}, \Name{check}, \Name{fail}, and \Name{outpos} corresponding to each array.
Figure \ref{fig:layout:pack} illustrates the \Name{Packed} layout.

\subsection{Array formats}
\label{sect:tech:format}

Double-array implementations often use 4-byte integers to represent an offset of arrays (e.g., \cite{darts,dartsclone}).
When \Name{BASE}, \Name{CHECK}, \Name{FAIL}, and \Name{OUTPOS} are implemented as arrays of 4-byte integers, we call this format \Name{Basic}.
In the following, we describe a more memory-efficient format that implements \Name{CHECK} as a byte array in the \Name{Bytewise} scheme, called \Name{Compact}.

\paragraph{Compact format}

The \Name{Compact} format stores transition labels in \Name{CHECK} instead of array offsets \cite{yata2007compact}.
In other words, $\Name{CHECK}[t]$ stores $c$ instead of $s$ when $\delta(s,c) = t$, and Equation \eqref{eq:da} is modified into
\begin{equation}
\Name{BASE}[s] + c = t ~ \textrm{and} ~ \Name{CHECK}[t] = c.
\end{equation}
However, for every state pair $(s,s')$, the following must be satisfied to avoid defining invalid transitions:
\begin{equation}
\label{eq:cda:base}
\Name{BASE}[s] \neq \Name{BASE}[s'].
\end{equation}

In the \Name{Bytewise} scheme, the alphabet size is bounded by 256, and \Name{CHECK}{} can be implemented as a byte array.
The \Name{Compact} format can reduce the memory consumption of \Name{CHECK} to 25\% while maintaining the time efficiency in matching.
In the \Name{Packed} layout, we assign three bytes to \Name{base}, one byte to \Name{check}, and four bytes to \Name{fail} and \Name{outpos} for avoiding extra memory padding on a 4-byte aligned memory structure, following the original implementation \cite{yata2007compact}.
This modification limits the maximum state id to $2^{24} - 1$. To represent state ids up to $2^{29} - 1$ with the same space usage, we can employ a technique used in the Darts-clone library \cite{dartsclone} that utilizes two types of offset values to represent exact and rough positions.
We did not use this technique in our experiments because three bytes are enough to store the automata we tested.

A drawback of \Name{Compact} is that it makes it harder for vacant searches to avoid the duplication of \Name{BASE} values for satisfying Equation \eqref{eq:cda:base};
indeed, this restriction can produce vacant ids that never satisfy Equation \eqref{eq:cda:base}.
For example, a vacant id $s$ never satisfies Equation \eqref{eq:cda:base} when \Name{BASE} values $s, s - 1,\dots, s - |\Sigma| + 1$ are already used.
In Section \ref{sect:exp:format}, we observe that such vacant ids can significantly slow the running times of vacant searches.

\subsection{Acceleration of vacant searches}
\label{sect:tech:constr}

Given a state having $k$ outgoing transitions with labels $c_1,c_2,\dots, c_k $, a na\"ive vacant search is performed to verify if $b + c_j$ for $1 \leq j \leq k$ are vacant ids for each integer $b \geq 0$ until finding such a value $b$.
However, for a maximum state id $N$, this approach verifies $O(N)$ integers in the worst case, resulting in slow construction for a large automaton.
In the following, we describe three acceleration techniques for vacant searches: \Name{Chain}, \Name{SkipForward}, and \Name{SkipDense}.

\paragraph{Chaining vacant ids}

The \Name{Chain} technique constructs a linked list on vacant ids \cite{morita2001fast}.
Let us denote $M$ vacant ids by $q_1,q_2,\dots,q_M$, where $0 < q_1 < q_2 < \dots < q_M < N$.
Given a state having $k$ outgoing transitions with labels $c_1,c_2,\dots, c_k$, \Name{Chain} verifies only \Name{BASE} values $b_i$ such that $b_i = q_i - \min\{ c_1, c_2, \dots, c_k \}$ for $i = 1,2,\dots,M$ while visiting only vacant ids using the linked list.
The number of verifications is bounded by $O(M)$.

If there are few vacant ids, $M$ becomes much smaller than $N$, and vacant searches can be performed faster than with the na\"ive approach.
The resultant double-array structure with \Name{Chain} is always identical to that with the na\"ive approach.

\paragraph{Skipping search blocks}

Before introducing \Name{SkipForward} and \Name{SkipDense}, we present a technique to partition array elements into fixed-size blocks \cite{dartsclone,yoshinaga2014self,kanda2017rearrangement}.
The partitioning modifies Equation \eqref{eq:da} into
\begin{equation}
\label{eq:da_xor}
\Name{BASE}[s] \oplus c = t ~ \textrm{and} ~ \Name{CHECK}[t] = s.
\end{equation}
In other words, the bitwise-XOR operation ($\oplus$) is used instead of the plus operation ($+$).

We introduce a constant parameter $B$ and partition array elements into blocks of size $B$, where the $s$-th element is placed in the $\Floor{s/B}$-th block.
We obtain the following theorem.

\begin{theorem}
Let $B = 2^{\Ceil{\log_2 |\Sigma|}}$.\footnote{In the \Name{Mapped} scheme, $\sigma$ is used instead of $|\Sigma|$.}
When state ids are defined using Equation \eqref{eq:da_xor}, all destination states from a state are always placed in the same block.
\end{theorem}
\begin{proof}
Computing $\Name{BASE}[s] \oplus c$ is implemented by modifying the least significant $\Ceil{\log_2 |\Sigma|}$ bits of $\Name{BASE}[s]$.
For all characters $c \in \Sigma$, the results of $\Name{BASE}[s] \oplus c$ have the same binary representation except the least significant $\Ceil{\log_2 |\Sigma|}$ bits.
Computing $\Floor{x/B}$ for an integer $x$ is implemented by shifting $\Ceil{\log_2 B} = \Ceil{\log_2 |\Sigma|}$ bits right, and the least significant $\Ceil{\log_2 |\Sigma|}$ bits of $x$ are omitted from the result.
Therefore, $\Floor{(\Name{BASE}[s] \oplus c )/B} = \Floor{\Name{BASE}[s]/B}$ for all characters $c \in \Sigma$, and all destination states $t$ from a state $s$ are placed in the $\Floor{\Name{BASE}[s]/B}$-th block. 
\end{proof}

\begin{figure}[tb]
\centering
\includegraphics[scale=0.9]{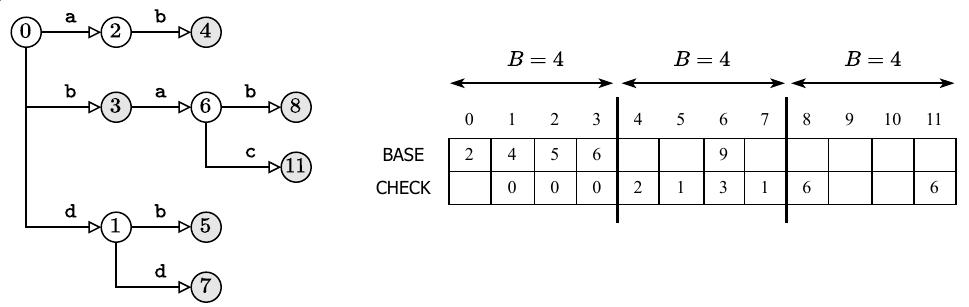}
\caption{
\Name{BASE} and \Name{CHECK} implementing the transition function $\delta$ of Figure \ref{fig:ac:trie} based on Equation \eqref{eq:da_xor}.
The block size $B$ is four, since $|\Sigma| = 4$ and $B = 2^{\Ceil{\log_2 4}}$.
Destination states 5 and 7 from state 1 are placed in the same block, i.e., $\Floor{5/4} = \Floor{7/4} = 1$.
}
\label{fig:vacant:xor}
\end{figure}

Figure \ref{fig:vacant:xor} shows an example of \Name{BASE} and \Name{CHECK} constructed using Equation \eqref{eq:da_xor}.
The partitioning allows us to perform vacant searches for blocks individually and set up conditions if we search each block.
\Name{SkipForward} and \Name{SkipDense} each use a different approach to select which blocks to be searched.

\begin{figure}[tb]
\centering
\subfloat[\Name{SkipForward}]{
\includegraphics[scale=0.9]{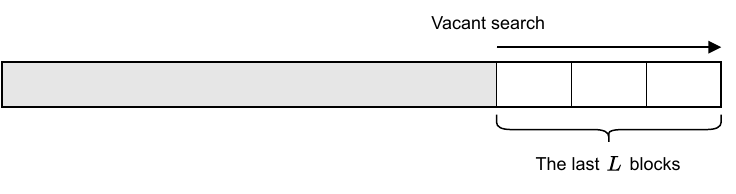}
\label{fig:vacant:t}
}\\
\subfloat[\Name{SkipDense}]{
\includegraphics[scale=0.9]{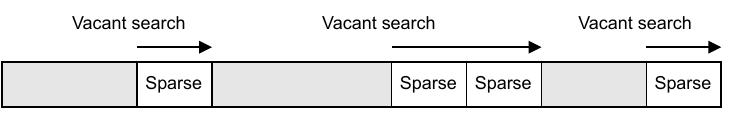}
\label{fig:vacant:s}
}
\caption{
Illustrations of skipping techniques.
}
\label{fig:vacant}
\end{figure}

\Name{SkipForward} searches for only the last $L$ blocks with a constant parameter $L$ \cite{dartsclone}, as shown in Figure \ref{fig:vacant:t}.
This idea is inspired by the fact that backward state ids are more likely to be vacant when performing vacant searches from the forward \cite{morita2001fast}.
\Name{SkipForward} can bound the number of verifications in a vacant search by $O(L \cdot |\Sigma|) = O(|\Sigma|)$, since $L$ is constant.
A drawback of \Name{SkipForward} is that blocks containing many vacant ids might be skipped, which degrades the memory efficiency.

\Name{SkipDense}\footnote{\Name{SkipDense} is inspired by the idea in the Darts library \cite{darts}, which skips forward elements in which the proportion of vacant ids is not less than a certain threshold. \Name{SkipDense} is a minor modification of the Darts technique combined with block partitioning and enables vacant searches to select target blocks more flexibly.} introduces a constant threshold $\tau \in [0,1]$ and categorizes blocks into two classes:
if the proportion of vacant ids in a block is no less than $\tau$, the block is categorized to the \emph{sparse} class;
otherwise, the block is categorized to the \emph{dense} class.
\Name{SkipDense} performs vacant searches for only sparse blocks, as shown in Figure \ref{fig:vacant:s}.
\Name{SkipDense} does not skip blocks that contain many vacant ids and is thus more memory-efficient than \Name{SkipForward}.
However, the number of verifications is not bounded by the alphabet size.

The ideas of both \Name{SkipForward} and \Name{SkipDense} do not conflict with \Name{Chain}.
We therefore apply the technique of \Name{Chain} to \Name{SkipForward} and \Name{SkipDense}.

\subsection{Traversal orders in construction}
\label{sect:tech:order}

\Name{BASE} and \Name{CHECK} are constructed by traversing an original trie from the root and performing vacant searches, as shown in Algorithm \ref{algo:constr}.
Although this algorithm traverses the original trie using a queue in the breadth-first order, we can take an arbitrary order to traverse the trie.
We can also take an arbitrary scan order of outgoing transitions for each visited state (i.e., the loop order of $E$ at Line \ref{algo:constr:line:edges} in Algorithm \ref{algo:constr}).

The orders are related to the resultant arrangement of states in \Name{BASE} and \Name{CHECK}, and the arrangement affects the cache efficiency in transitions:
a transition $\delta(s,c) = t$ can be looked up quickly when states $s$ and $t$ are placed close together (e.g., on the same cache line).
Selecting the traversal order that improves the locality of reference during a matching is crucial.
In the following, we describe the four approaches: \Name{LexBFS}, \Name{LexDFS}, \Name{FreqBFS}, and \Name{FreqDFS}.

\paragraph{Breadth- and depth-first searches}

There are two types of data structures to traverse the original trie.
One uses a queue and visits states with a breadth-first search (as Algorithm \ref{algo:constr}), which we denote by \Name{BFS}.
The other uses a stack and visits states with a depth-first search, which we denote by \Name{DFS}.

We consider performing vacant searches from the forward (i.e., attempting to use smaller state ids first). The resultant arrangements of states using \Name{BFS} and \Name{DFS} are expected to be as follows.
With \Name{BFS}, it is expected that shallower states (i.e., ones around the initial state) are placed around the head of \Name{BASE} and \Name{CHECK}.
If we often visit shallow states during a matching, \Name{BFS} can perform cache-efficiently.
With \Name{DFS}, in contrast, it is expected that deeper states are placed close together.
If we often visit deep states during a matching, \Name{DFS} can perform cache-efficiently.

\begin{figure}[tbp]
\centering
\subfloat[\Name{BFS}]{
\includegraphics[scale=0.9]{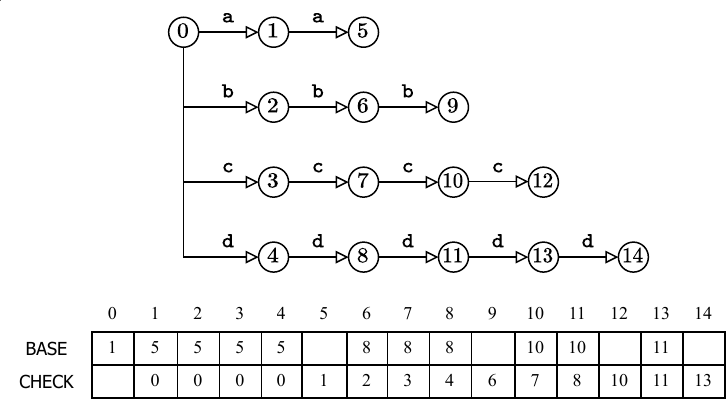}
\label{fig:ex:order:lexbfs}
}\\
\subfloat[\Name{DFS}]{
\includegraphics[scale=0.9]{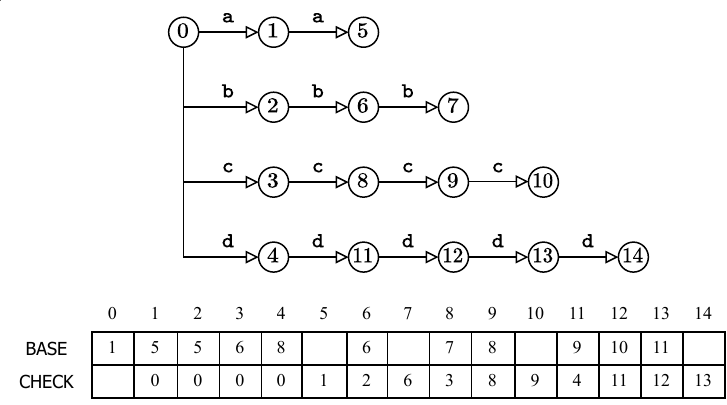}
\label{fig:ex:order:lexdfs}
}
\caption{
Examples of resultant \Name{BASE} and \Name{CHECK} arrays constructed with \Name{BFS} and \Name{DFS} for a dictionary of four patterns \Str{aa}, \Str{bbb}, \Str{cccc}, and \Str{ddddd}.
$\Sigma = \{ \Str{a} = 0,\Str{b} = 1,\Str{c} = 2,\Str{d} = 3 \}$.
The state ids are determined based on Equation \eqref{eq:da}.
In the construction, vacant searches are performed from the forward to use smaller state ids first.
A set of outgoing transition labels (or $E$) is scanned in the lexicographical order (i.e., the figures show the results of \Name{LexBFS} and \Name{LexDFS} more precisely).
In both cases, the initial state and its destination states have identical ids (i.e., states 0 to 4).
This is because the initial state id is always fixed to 0, and its destination state ids are fixed to $\Name{BASE}[0]$ to satisfy Equation \eqref{eq:da}.
The deeper states have different ids in accordance with the order of visiting states.
}
\label{fig:ex:order}
\end{figure}

Figure \ref{fig:ex:order} shows toy examples of resultant \Name{BASE} and \Name{CHECK} arrays with \Name{BFS} and \Name{DFS}.
In \Name{BFS}, shallower states have smaller state ids (see Figure \ref{fig:ex:order:lexbfs}), while
in \Name{DFS}, states whose depth is no less than two have adjacent ids with their destination states (see Figure \ref{fig:ex:order:lexdfs}).
These examples show that \Name{BFS} and \Name{DFS} have better locality of reference when we visit shallower and deeper states, respectively.

\paragraph{Lexicographical and frequency orders}

There are two orders to traverse outgoing transitions for each visited state: 
\Name{Lex} visits transitions in the lexicographical order of their labels, and
\Name{Freq} visits transitions in the decreasing order of the frequencies of their labels in $\mathcal{D}$.
In the \Name{Freq} order, ids of states incoming with frequent transition labels are determined earlier, and it is expected that those states are placed at the head of the array.

\paragraph{Possible combinations}

There are four possible combinations in the traversal order:
\begin{itemize}
\item \Name{LexBFS} is \Name{BFS} with \Name{Lex},
\item \Name{FreqBFS} is \Name{BFS} with \Name{Freq},
\item \Name{LexDFS} is \Name{DFS} with \Name{Lex}, and
\item \Name{FreqDFS} is \Name{DFS} with \Name{Freq}.
\end{itemize}

\subsection{Summary of implementation techniques}

Table \ref{tab:summary} summarizes the implementation techniques described in this section.
Both \Name{Bytewise} and \Name{Mapped} are checked in the ``Selected'' column because the \Name{Mapped} scheme is optimized for Asian languages (such as Japanese, Chinese, and Korean), which deal with large character sets, while the \Name{Bytewise} scheme performs better in other languages with smaller character sets.

\begin{table}[tbp]
\footnotesize
\centering
\caption{
Summary of implementation techniques in DAACs.
``Selected'' indicates the techniques selected to be the best through our experiments in Section \ref{sect:exp} and are used in our Daachorse library.
}
\label{tab:summary}
\subfloat[Approaches to store output sets (Section \ref{sect:tech:out})]{
\begin{tabular}{|c||l|c|c|}
\hline
Technique & Brief description & References  & Selected \\
\hline\hline
\Name{Simple} & Arranging pattern indices in a simple manner & Conventional & \\ \hline
\Name{Shared} & Merging some common parts of pattern indices in \Name{Simple} & This study & \\ \hline
\Name{Forest} & Storing pattern indices in a forest structure & \cite{belazzougui2010succinct} & \checkmark{} \\ \hline
\end{tabular}
}\\
\subfloat[Schemes to handle strings of multibyte characters (Section \ref{sect:tech:wise})]{
\begin{tabular}{|c||l|c|c|}
\hline
Technique & Brief description & References & Selected \\
\hline\hline
\Name{Bytewise} & Slicing multibyte characters into byte sequences & Conventional & \checkmark{} \\ \hline
\Name{Charwise} & Handling multibyte characters as code points in Unicode & Conventional & \\ \hline
\Name{Mapped} & Mapping code points in the frequency order & \cite{liu2011compression,norimatsu2016fast} & \checkmark{} \\ \hline
\end{tabular}
}\\
\subfloat[Memory layouts of double-arrays (Section \ref{sect:tech:layout})]{
\begin{tabular}{|c||l|c|c|}
\hline
Technique & Brief description & References & Selected \\
\hline\hline
\Name{Individual} & Maintaining arrays \Name{BASE}, \Name{CHECK}, \Name{FAIL}, and \Name{OUTPOS} individually & Conventional & \\ \hline
\Name{Packed} & Arranging values in the arrays cache-efficiently & e.g., \cite{morita2001fast,yata2007compact} & \checkmark{} \\ \hline
\end{tabular}
}\\
\subfloat[Formats of double-arrays in \Name{Bytewise} scheme (Section \ref{sect:tech:format})]{
\begin{tabular}{|c||l|c|c|}
\hline
Technique & Brief description & References & Selected \\
\hline\hline
\Name{Basic} & Implementing \Name{BASE}, \Name{CHECK}, \Name{FAIL}, and \Name{OUTPOS} as arrays of four-byte integers & Conventional & \\ \hline
\Name{Compact} & Compressing \Name{CHECK} into a byte array  & \cite{yata2007compact} & \checkmark{} \\ \hline
\end{tabular}
}\\
\subfloat[Approaches to accelerate vacant searches (Section \ref{sect:tech:constr})]{
\begin{tabular}{|c||l|c|c|}
\hline
Technique & Brief description & References & Selected \\
\hline\hline
\Name{Chain} & Visiting only vacant ids with a linked list & \cite{morita2001fast} & \\ \hline
\Name{SkipForward} & Searching for only the last $L$ blocks & \cite{morita2001fast,dartsclone} & \checkmark{} \\ \hline
\Name{SkipDense} & Searching for only blocks in which the vacant proportion is no less than $\tau$ & This study & \\ \hline
\end{tabular}
}\\
\subfloat[Traversal orders in construction (Section \ref{sect:tech:order})]{
\begin{tabular}{|c||l|c|c|}
\hline
Technique & Brief description & References & Selected \\
\hline\hline
\Name{LexBFS} & Visiting states with breadth-first search and transition labels in the lexicographical order & Conventional & \\ \hline
\Name{FreqBFS} & Visiting states with breadth-first search and transition labels in the frequency order & Conventional & \\ \hline
\Name{LexDFS} & Visiting states with depth-first search and transition labels in the lexicographical order & Conventional & \checkmark{} \\ \hline
\Name{FreqDFS} & Visiting states with depth-first search and transition labels in the frequency order & Conventional & \\ \hline
\end{tabular}
}
\end{table}

\section{Experimental analyses of DAAC techniques}
\label{sect:exp}

\subsection{Setup}

We conducted all experiments on one core of a hexa-core Intel i7-8086K CPU clocked at 2.80 GHz in a machine with 64 GB of RAM (L1 cache: 32 KiB, L2 cache: 256 KiB, L3 cache: 12 MiB), running the 64-bit version of CentOS 7.5 based on Linux 3.10. 
All data structures were implemented in Rust.
We compiled the source code by rustc 1.60.0 with optimization flag \texttt{opt-level=3}.

\paragraph{Datasets}

We use three real-world natural language datasets:

\begin{itemize}
\item \Name{EnWord} is English word uni-grams in the Google Web 1T 5-Gram Corpus \cite{brants2006web};
\item \Name{JaWord} is Japanese word uni-grams in the Nihongo Web Corpus 2010 \cite{yata2010nwc}; and
\item \Name{JaChars} is Japanese character uni-, bi-, and tri-grams in the Nihongo Web Corpus 2010 \cite{yata2010nwc}.
\end{itemize}

\begin{table}[btp]
\footnotesize
\centering
\caption{Basic statistics of dictionaries.}
\label{tab:dataset}
\subfloat[Average pattern length in bytes]{
\label{tab:dataset:lenbyte}
\begin{tabular}{|l||rrrr|}
\hline
Dataset & 1K & 10K & 100K & 1M \\
\hline\hline
\Name{EnWord} & 5.0 & 6.6 & 7.1 & 7.5 \\
\Name{JaWord} & 5.7 & 6.7 & 8.7 & 14.5 \\
\Name{JaChars} & 4.9 & 6.3 & 7.3 & 8.0 \\
\hline
\end{tabular}
}
\subfloat[Average pattern length in characters]{
\label{tab:dataset:lenchar}
\begin{tabular}{|l||rrrr|}
\hline
Dataset & 1K & 10K & 100K & 1M \\
\hline\hline
\Name{EnWord} & 5.0 & 6.6 & 7.1 & 7.5 \\
\Name{JaWord} & 1.9 & 2.3 & 3.0 & 5.2 \\
\Name{JaChars} & 1.8 & 2.3 & 2.7 & 2.9 \\
\hline
\end{tabular}
}
\subfloat[Alphabet size in the \Name{Bytewise} scheme]{
\label{tab:dataset:alphbyte}
\begin{tabular}{|l||rrrr|}
\hline
Dataset & 1K & 10K & 100K & 1M \\
\hline\hline
\Name{EnWord} & 226 & 226 & 239 & 239 \\
\Name{JaWord} & 233 & 238 & 239 & 244 \\
\Name{JaChars} & 233 & 233 & 238 & 239 \\
\hline
\end{tabular}
}\\
\subfloat[Alphabet size in the \Name{Charwise} scheme]{
\label{tab:dataset:alphchar}
\begin{tabular}{|l||rrrr|}
\hline
Dataset & 1K & 10K & 100K & 1M \\
\hline\hline
\Name{EnWord} & 8,226 & 9,632 & 63,289 & 65,532 \\
\Name{JaWord} & 39,640 & 57,345 & 64,017 & 1,048,766 \\
\Name{JaChars} & 39,640 & 40,845 & 57,345 & 64,017 \\
\hline
\end{tabular}
}
\subfloat[Mapped alphabet size in the \Name{Mapped} scheme]{
\label{tab:dataset:alphcharm}
\begin{tabular}{|l||rrrr|}
\hline
Dataset & 1K & 10K & 100K & 1M \\
\hline\hline
\Name{EnWord} & 83 & 104 & 177 & 472 \\
\Name{JaWord} & 669 & 2,262 & 4,806 & 10,809 \\
\Name{JaChars} & 483 & 1,648 & 2,990 & 4,773 \\
\hline
\end{tabular}
}
\end{table}

In these datasets, each N-gram has its own frequency.
We extract the most frequent $K$ N-grams from each dataset and produce dictionaries $\mathcal{D}$ of $K$ patterns.
We test $K = 10^3, 10^4, 10^5, 10^6$ to observe scalability.
Table \ref{tab:dataset} lists the basic statistics of the dictionaries.

\begin{table}[btp]
\footnotesize
\centering
\caption{The number of occurrences in pattern matching per sentence.}
\label{tab:dataset:occur}
\begin{tabular}{|l||rrrr|}
\hline
Dataset & 1K & 10K & 100K & 1M \\
\hline\hline
\Name{EnWord} & 47.2 & 93.4 & 126.2 & 150.8 \\
\Name{JaWord} & 25.8 & 38.0 & 46.7 & 49.0 \\
\Name{JaChars} & 36.6 & 54.6 & 69.7 & 79.9 \\
\hline
\end{tabular}
\end{table}

We evaluate the running times of pattern matching and the construction and memory usage of an AC automaton.
We use an English text from the Pizza\&Chili Corpus \cite{pctext} and a Japanese text from 13 different subcorpora in the Balanced Corpus of Contemporary Written Japanese (BCCWJ, version 1.1) \cite{maekawa2014bccwj} to measure pattern matching times for the English and Japanese dictionaries, respectively.
We randomly sample one million sentences (separated by lines) from each text and produce a sequence of search texts $T$.
Table \ref{tab:dataset:occur} shows the number of occurrences during a pattern matching for each dataset.
Throughout the experiments, we report the total running time of a pattern matching for all sentences, averaged on ten runs.

\subsection{Analysis on approaches to store output sets}

Table \ref{tab:output} shows the experimental results on the \Name{Simple}, \Name{Shared}, and \Name{Forest} approaches presented in Section \ref{sect:tech:out}, while fixing the other settings to \Name{Bytewise}, \Name{Packed}, \Name{Basic}, \Name{Chain}, and \Name{LexDFS}.

\begin{table}[btp]
\footnotesize
\centering
\caption{
Experimental results for approaches to store output sets: \Name{Simple}, \Name{Shared}, and \Name{Forest}.
For memory usage, the total of \Name{OUTPUT} and \Name{TERM} is reported in \Name{Simple} and \Name{Shared}, and the total of \Name{OUTPUT} and \Name{PARENT} is reported in \Name{Forest}.
The last two rows in Tables (a)--(c) show the different ratios for each result.
}
\label{tab:output}
\subfloat[Length of \Name{OUTPUT} ($\times 10^3$)]{
\label{tab:output:length}
\begin{tabular}{|l||rrrr|rrrr|rrrr|}
\hline
 & \multicolumn{4}{c|}{\Name{EnWord}} & \multicolumn{4}{c|}{\Name{JaWord}} & \multicolumn{4}{c|}{\Name{JaChars}} \\
Approach & 1K & 10K & 100K & 1M & 1K & 10K & 100K & 1M & 1K & 10K & 100K & 1M \\
\hline\hline
\Name{Simple} & 1.60 & 29.2 & 412 & 4,988 & 1.46 & 18.7 & 247 & 3,615 & 1.68 & 21.6 & 252 & 2,779 \\
\Name{Shared} & 1.46 & 26.2 & 362 & 4,377 & 1.31 & 16.7 & 218 & 3,132 & 1.40 & 17.7 & 217 & 2,497 \\
\Name{Forest} & 1.00 & 10.0 & 100 & 1,000 & 1.00 & 10.0 & 100 & 1,000 & 1.00 & 10.0 & 100 & 1,000 \\
\hline
\Name{Shared}/\Name{Simple} & 0.92 & 0.90 & 0.88 & 0.88 & 0.89 & 0.89 & 0.88 & 0.87 & 0.83 & 0.82 & 0.86 & 0.90 \\
\Name{Forest}/\Name{Simple} & 0.63 & 0.34 & 0.24 & 0.20 & 0.68 & 0.53 & 0.40 & 0.28 & 0.60 & 0.46 & 0.40 & 0.36 \\
\hline
\end{tabular}
}\\
\subfloat[Memory usage (KiB)]{
\label{tab:output:memory}
\begin{tabular}{|l||rrrr|rrrr|rrrr|}
\hline
 & \multicolumn{4}{c|}{\Name{EnWord}} & \multicolumn{4}{c|}{\Name{JaWord}} & \multicolumn{4}{c|}{\Name{JaChars}} \\
Approach & 1K & 10K & 100K & 1M & 1K & 10K & 100K & 1M & 1K & 10K & 100K & 1M \\
\hline\hline
\Name{Simple} & 12.5 & 228 & 3,219 & 38,968 & 11.4 & 146 & 1,932 & 28,244 & 13.1 & 168 & 1,971 & 21,707 \\
\Name{Shared} & 11.4 & 205 & 2,830 & 34,192 & 10.2 & 131 & 1,703 & 24,473 & 10.9 & 139 & 1,692 & 19,507 \\
\Name{Forest} & 11.7 & 117 & 1,172 & 11,719 & 11.7 & 117 & 1,172 & 11,719 & 11.7 & 117 & 1,172 & 11,719 \\
\hline
\Name{Shared}/\Name{Simple} & 0.92 & 0.90 & 0.88 & 0.88 & 0.89 & 0.89 & 0.88 & 0.87 & 0.83 & 0.82 & 0.86 & 0.90 \\
\Name{Forest}/\Name{Simple} & 0.94 & 0.51 & 0.36 & 0.30 & 1.02 & 0.80 & 0.61 & 0.41 & 0.89 & 0.70 & 0.59 & 0.54 \\
\hline
\end{tabular}
}\\
\subfloat[Matching time (ms)]{
\label{tab:output:search}
\begin{tabular}{|l||rrrr|rrrr|rrrr|}
\hline
 & \multicolumn{4}{c|}{\Name{EnWord}} & \multicolumn{4}{c|}{\Name{JaWord}} & \multicolumn{4}{c|}{\Name{JaChars}} \\
Approach & 1K & 10K & 100K & 1M & 1K & 10K & 100K & 1M & 1K & 10K & 100K & 1M \\
\hline\hline
\Name{Simple} & 525 & 609 & 693 & 1,088 & 469 & 646 & 919 & 1,266 & 448 & 700 & 1,073 & 2,853 \\
\Name{Shared} & 525 & 613 & 719 & 1,148 & 457 & 653 & 923 & 1,304 & 415 & 668 & 1,044 & 2,822 \\
\Name{Forest} & 533 & 596 & 697 & 1,088 & 446 & 626 & 871 & 1,259 & 427 & 663 & 1,027 & 2,796 \\
\hline
\Name{Shared}/\Name{Simple} & 1.00 & 1.01 & 1.04 & 1.06 & 0.98 & 1.01 & 1.00 & 1.03 & 0.92 & 0.95 & 0.97 & 0.99 \\
\Name{Forest}/\Name{Simple} & 1.01 & 0.98 & 1.00 & 1.00 & 0.95 & 0.97 & 0.95 & 0.99 & 0.95 & 0.95 & 0.96 & 0.98 \\
\hline
\end{tabular}
}\\
\subfloat[Average cardinality of $h(s)$]{
\label{tab:output:card}
\begin{tabular}{|l||rrrr|}
\hline
Dataset & 1K & 10K & 100K & 1M \\
\hline\hline
\Name{EnWord} & 1.37 & 2.07 & 3.16 & 4.15 \\
\Name{JaWord} & 1.42 & 1.76 & 2.20 & 3.02 \\
\Name{JaChars} & 1.63 & 2.06 & 2.42 & 2.73 \\
\hline
\end{tabular}
}
\end{table}

We first focus on the statistics related to memory efficiency.
Table \ref{tab:output:length} shows the length of \Name{OUTPUT}, i.e., the number of pattern indices stored in the data structure.
The length in \Name{Forest} is essentially the same as the number of patterns and is much smaller than that in \Name{Simple}, indicating that \Name{Simple} maintains many duplicate values in \Name{OUTPUT}.
\Name{Shared} also reduces such duplicate values in \Name{Simple}.
Comparing \Name{Shared} and \Name{Forest} on their reduction ratios from \Name{Simple}, the difference becomes larger as the number of patterns increases, and \Name{Forest} is more memory efficient for large dictionaries.
Table \ref{tab:output:memory} shows the memory usage of the data structures.
\Name{Shared} is the smallest when the number of patterns is 1K, and \Name{Forest} is the smallest in the other cases.

We next focus on the statistics related to the time efficiency of the AC algorithm.
Table \ref{tab:output:search} reports the elapsed time during a pattern matching, and as we can see, there is no significant difference between the approaches.
These results indicate that the cache-efficient scanning in \Name{Simple} and \Name{Shared} is unimportant.
To clarify this, Table \ref{tab:output:card} shows the average cardinality of $h(s)$ for each dictionary.
These results imply that only a few random accesses can arise in \Name{Forest}'s extraction and do not create a bottleneck.

\subsection{Analysis on byte-wise and character-wise schemes}
\label{sect:exp:wise}

We evaluate the performances of the \Name{Bytewise}, \Name{Charwise}, and \Name{Mapped} schemes presented in Section \ref{sect:tech:wise}, while fixing the other settings to \Name{Forest}, \Name{Packed}, \Name{Basic}, \Name{Chain}, and \Name{LexDFS}.

\begin{table}[btp]
\footnotesize
\centering
\caption{
Experimental results on memory efficiency in the \Name{Bytewise}, \Name{Charwise}, and \Name{Mapped} schemes.
The memory usage is the total of \Name{BASE}, \Name{CHECK}, \Name{FAIL}, and \Name{OUTPOS} (and the mapping $\pi$ in \Name{Mapped}).
The last rows in Tables (a) and (d) show the different ratios for each result.
}
\label{tab:scheme:stat}
\subfloat[Number of original states ($\times 10^3$)]{
\label{tab:scheme:states}
\begin{tabular}{|l||rrrr|rrrr|rrrr|}
\hline
 & \multicolumn{4}{c|}{\Name{EnWord}} & \multicolumn{4}{c|}{\Name{JaWord}} & \multicolumn{4}{c|}{\Name{JaChars}} \\
Scheme & 1K & 10K & 100K & 1M & 1K & 10K & 100K & 1M & 1K & 10K & 100K & 1M \\
\hline\hline
\Name{Bytewise} & 2.74 & 25.6 & 237 & 2,129 & 2.96 & 27.9 & 305 & 4,899 & 1.73 & 18.5 & 186 & 1,790 \\
\Name{Charwise} & 2.74 & 25.6 & 237 & 2,124 & 1.43 & 13.0 & 142 & 2,155 & 1.12 & 10.9 & 107 & 1,052 \\
\hline
\Name{Charwise}/\Name{Bytewise} & 1.00 & 1.00 & 1.00 & 1.00 & 0.48 & 0.47 & 0.46 & 0.44 & 0.65 & 0.59 & 0.58 & 0.59 \\
\hline
\end{tabular}
}\\
\subfloat[Proportion of vacant ids (\%)]{
\label{tab:scheme:vacant}
\begin{tabular}{|l||rrrr|rrrr|rrrr|}
\hline
 & \multicolumn{4}{c|}{\Name{EnWord}} & \multicolumn{4}{c|}{\Name{JaWord}} & \multicolumn{4}{c|}{\Name{JaChars}} \\
Scheme & 1K & 10K & 100K & 1M & 1K & 10K & 100K & 1M & 1K & 10K & 100K & 1M \\
\hline\hline
\Name{Bytewise} & 2.6 & 0.9 & 0.1 & 0.0 & 3.7 & 0.8 & 0.2 & 0.1 & 15.3 & 0.8 & 0.2 & 0.0 \\
\Name{Charwise} & 66.6 & 0.2 & 0.0 & 0.0 & 96.4 & 77.3 & 16.3 & 0.7 & 97.2 & 84.1 & 68.0 & 71.8 \\
\Name{Mapped} & 2.7 & 0.9 & 0.2 & 0.0 & 30.2 & 20.5 & 3.9 & 1.1 & 27.1 & 24.1 & 12.6 & 27.0 \\
\hline
\end{tabular}
}\\
\subfloat[Average number of outgoing transitions for an internal state.]{
\label{tab:scheme:outgoing}
\begin{tabular}{|l||rrrr|rrrr|rrrr|}
\hline
 & \multicolumn{4}{c|}{\Name{EnWord}} & \multicolumn{4}{c|}{\Name{JaWord}} & \multicolumn{4}{c|}{\Name{JaChars}} \\
Scheme & 1K & 10K & 100K & 1M & 1K & 10K & 100K & 1M & 1K & 10K & 100K & 1M \\
\hline\hline
\Name{Bytewise} & 1.41 & 1.40 & 1.44 & 1.50 & 1.41 & 1.43 & 1.37 & 1.20 & 1.87 & 1.68 & 1.75 & 1.91 \\
\Name{Charwise} & 1.41 & 1.41 & 1.44 & 1.50 & 2.54 & 2.79 & 2.40 & 1.60 & 3.58 & 3.25 & 3.83 & 5.33 \\
\hline
\end{tabular}
}\\
\subfloat[Memory usage (KiB)]{
\label{tab:scheme:memory}
\begin{tabular}{|l||rrrr|rrrr|rrrr|}
\hline
 & \multicolumn{4}{c|}{\Name{EnWord}} & \multicolumn{4}{c|}{\Name{JaWord}} & \multicolumn{4}{c|}{\Name{JaChars}} \\
Scheme & 1K & 10K & 100K & 1M & 1K & 10K & 100K & 1M & 1K & 10K & 100K & 1M \\
\hline\hline
\Name{Bytewise} & 44 & 404 & 3,712 & 33,276 & 48 & 440 & 4,780 & 76,596 & 32 & 292 & 2,908 & 27,972 \\
\Name{Charwise} & 128 & 401 & 3,707 & 33,193 & 619 & 896 & 2,646 & 33,922 & 619 & 1,067 & 5,240 & 58,318 \\
\Name{Mapped} & 76 & 442 & 3,959 & 33,456 & 187 & 480 & 2,554 & 38,145 & 179 & 384 & 2,144 & 22,778 \\
\hline
\Name{Charwise}/\Name{Bytewise} & 2.91 & 0.99 & 1.00 & 1.00 & 12.89 & 2.04 & 0.55 & 0.44 & 19.34 & 3.65 & 1.80 & 2.08 \\
\Name{Mapped}/\Name{Bytewise} & 1.73 & 1.09 & 1.07 & 1.01 & 3.89 & 1.09 & 0.53 & 0.50 & 5.59 & 1.31 & 0.74 & 0.81 \\
\Name{Mapped}/\Name{Charwise} & 0.59 & 1.10 & 1.07 & 1.01 & 0.30 & 0.54 & 0.97 & 1.12 & 0.29 & 0.36 & 0.41 & 0.39 \\
\hline
\end{tabular}
}
\end{table}

Table \ref{tab:scheme:stat} shows the results on memory efficiency,
where Table \ref{tab:scheme:states} lists the number of original states (i.e., $|S|$) and
Table \ref{tab:scheme:vacant} reports the proportion of vacant ids (i.e., $|S|/|S_{\Name{BC}}|$), which we call \emph{vacant proportion}.
On \Name{JaWord} and \Name{JaChars}, whose strings consist of multibyte characters, \Name{Charwise} defines fewer states than \Name{Bytewise} but more vacant ids.
Especially in \Name{JaChars}, the vacant proportion in \Name{Charwise} is always more than 68\%.
The large vacant proportions are related to the average number of outgoing transitions from an internal state, as reported in Table \ref{tab:scheme:outgoing}.
The more outgoing transitions there are, the easier it is for vacant searches to fail and the number of vacant ids to increase.
The large vacant proportions of \Name{Charwise}, however, can be improved with the code mapping $\pi$, as the result of \Name{Mapped} demonstrates.

Table \ref{tab:scheme:memory} reports the total memory usage of \Name{BASE}, \Name{CHECK}, \Name{FAIL}, and \Name{OUTPOS} (and the mapping $\pi$ in \Name{Mapped}).
On \Name{EnWord}, there is no significant difference between the schemes when the number of patterns is no less than 10K.
On \Name{JaWord} and \Name{JaChars}, \Name{Mapped} is always the first or second smallest because of its fewer states and smaller vacant proportions.

\begin{table}[btp]
\footnotesize
\centering
\caption{
Experimental results on time efficiency in the \Name{Bytewise}, \Name{Charwise}, and \Name{Mapped} schemes.
The last rows in Tables (a) and (c) show the different ratios for each result.
}
\label{tab:scheme}
\subfloat[Number of visiting states during matching ($\times 10^6$)]{
\label{tab:scheme:trans}
\begin{tabular}{|l||rrrr|rrrr|rrrr|}
\hline
 & \multicolumn{4}{c|}{\Name{EnWord}} & \multicolumn{4}{c|}{\Name{JaWord}} & \multicolumn{4}{c|}{\Name{JaChars}} \\
Scheme & 1K & 10K & 100K & 1M & 1K & 10K & 100K & 1M & 1K & 10K & 100K & 1M \\
\hline\hline
\Name{Bytewise} & 108 & 104 & 107 & 111 & 128 & 128 & 129 & 129 & 129 & 129 & 129 & 131 \\
\Name{Charwise} & 108 & 104 & 107 & 111 & 57 & 61 & 61 & 61 & 59 & 62 & 62 & 62 \\
\hline
\Name{Charwise}/\Name{Bytewise} & 1.00 & 1.00 & 1.00 & 1.00 & 0.44 & 0.48 & 0.47 & 0.47 & 0.45 & 0.48 & 0.48 & 0.48 \\
\hline
\end{tabular}
}\\
\subfloat[Matching time (ms)]{
\label{tab:scheme:search}
\begin{tabular}{|l||rrrr|rrrr|rrrr|}
\hline
 & \multicolumn{4}{c|}{\Name{EnWord}} & \multicolumn{4}{c|}{\Name{JaWord}} & \multicolumn{4}{c|}{\Name{JaChars}} \\
Scheme & 1K & 10K & 100K & 1M & 1K & 10K & 100K & 1M & 1K & 10K & 100K & 1M \\
\hline\hline
\Name{Bytewise} & 527 & 597 & 689 & 1,057 & 454 & 644 & 896 & 1,254 & 406 & 648 & 1,021 & 2,735 \\
\Name{Charwise} & 550 & 628 & 720 & 1,098 & 362 & 444 & 563 & 809 & 356 & 507 & 661 & 1,812 \\
\Name{Mapped} & 576 & 665 & 748 & 1,157 & 310 & 402 & 484 & 731 & 304 & 439 & 620 & 1,403 \\
\hline
\Name{Charwise}/\Name{Bytewise} & 1.04 & 1.05 & 1.04 & 1.04 & 0.80 & 0.69 & 0.63 & 0.65 & 0.88 & 0.78 & 0.65 & 0.66 \\
\Name{Mapped}/\Name{Bytewise} & 1.09 & 1.11 & 1.09 & 1.09 & 0.68 & 0.62 & 0.54 & 0.58 & 0.75 & 0.68 & 0.61 & 0.51 \\
\Name{Mapped}/\Name{Charwise} & 1.05 & 1.06 & 1.04 & 1.05 & 0.86 & 0.91 & 0.86 & 0.90 & 0.85 & 0.87 & 0.94 & 0.77 \\
\hline
\end{tabular}
}\\
\subfloat[Construction time (ms)]{
\label{tab:scheme:build}
\begin{tabular}{|l||rrrr|rrrr|rrrr|}
\hline
 & \multicolumn{4}{c|}{\Name{EnWord}} & \multicolumn{4}{c|}{\Name{JaWord}} & \multicolumn{4}{c|}{\Name{JaChars}} \\
Scheme & 1K & 10K & 100K & 1M & 1K & 10K & 100K & 1M & 1K & 10K & 100K & 1M \\
\hline\hline
\Name{Bytewise} & 0.12 & 1.34 & 21 & 244 & 0.13 & 1.44 & 27 & 387 & 0.09 & 1.2 & 20 & 275 \\
\Name{Charwise} & 0.14 & 0.79 & 18 & 199 & 0.32 & 1.09 & 163 & 1,829 & 0.29 & 2.3 & 134 & 21,184 \\
\Name{Mapped} & 0.09 & 0.81 & 18 & 204 & 0.06 & 1.50 & 102 & 1,251 & 0.08 & 1.4 & 45 & 9,107 \\
\hline
\Name{Charwise}/\Name{Bytewise} & 1.18 & 0.59 & 0.86 & 0.81 & 2.51 & 0.76 & 6.03 & 4.72 & 3.14 & 1.92 & 6.72 & 77.1 \\
\Name{Mapped}/\Name{Bytewise} & 0.80 & 0.60 & 0.86 & 0.83 & 0.44 & 1.04 & 3.77 & 3.23 & 0.83 & 1.22 & 2.25 & 33.1 \\
\Name{Mapped}/\Name{Charwise} & 0.68 & 1.03 & 1.01 & 1.03 & 0.17 & 1.37 & 0.62 & 0.68 & 0.26 & 0.64 & 0.33 & 0.43 \\
\hline
\end{tabular}
}
\end{table}

Table \ref{tab:scheme} shows the experimental results on time efficiency, where
Table \ref{tab:scheme:trans} reports the number of visiting states during matching with $\delta$ and $f$.
On \Name{EnWord}, whose characters are mostly single-byte ones, there is no significant difference between the schemes.
On \Name{JaWord} and \Name{JaChars}, whose strings consist of multibyte characters, \Name{Charwise} (or \Name{Mapped}) achieves half as many as \Name{Bytewise}, resulting in halving the number of random accesses during a pattern matching.
Table \ref{tab:scheme:search} lists the elapsed time during a pattern matching.
On \Name{JaWord} and \Name{JaChars}, \Name{Charwise} and  \Name{Mapped} are faster than \Name{Bytewise} because of their fewer random accesses.
Comparing \Name{Charwise} and \Name{Mapped}, there is no large difference, although \Name{Mapped} requires additional computations for the code mapping $\pi$.

Table \ref{tab:scheme:build} shows the elapsed time to compute two arrays, \Name{BASE} and \Name{CHECK}, from an original trie.
On \Name{EnWord}, \Name{Charwise} and \Name{Mapped} are faster than \Name{Bytewise} in most cases.
On \Name{JaWord} and \Name{JaChars}, \Name{Charwise} is much slower than \Name{Bytewise}, and its time performance is improved by the code mapping of \Name{Mapped};
nevertheless, \Name{Mapped} is still much slower than \Name{Bytewise} for large dictionaries.
The time inefficiency of \Name{Charwise} and \Name{Mapped} is caused by their large vacant proportion (as shown in Table \ref{tab:scheme:vacant}):
\Name{Chain} does not perform well when there are many vacant ids.

\paragraph{Memory layouts of the $\pi$ data structure}

We further investigate the data structure of $\pi$.
In the other experiments, we use a simple four-byte array to implement the $\pi$ data structure.
On the one hand, it can also be implemented with a three-byte array because the number of code points is bounded by $2^{21}$.
Although the three-byte array implementation can save the memory consumption, extracting $\pi(c)$ will take more CPU instructions to convert a byte sequence into a four-byte integer type.

\begin{table}[btp]
\footnotesize
\centering
\caption{
Experimental results on memory efficiency, matching time, and CPU instructions when varying the memory layout of the $\pi$ data structure in the \Name{Mapped} scheme.
\Name{4B} and \Name{3B} indicate the cases that $\pi$ is implemented with four- and three-byte arrays, respectively.
The memory usage is the total of \Name{BASE}, \Name{CHECK}, \Name{FAIL}, \Name{OUTPOS}, and $\pi$.
The last rows in each table show the different ratios for each result.
}
\label{tab:scheme:3B}
\subfloat[Memory usage (KiB)]{
\label{tab:scheme:3B:memory}
\begin{tabular}{|l||rrrr|rrrr|}
\hline
  & \multicolumn{4}{c|}{\Name{JaWord}} & \multicolumn{4}{c|}{\Name{JaChars}} \\
Layout & 1K & 10K & 100K & 1M & 1K & 10K & 100K & 1M \\
\hline\hline
\Name{4B} & 187 & 480 & 2,554 & 38,145 & 179 & 384 & 2,144 & 22,778 \\
\Name{3B} & 148 & 424 & 2,492 & 37,121 & 140 & 344 & 2,088 & 22,716 \\
\hline
\Name{3B}/\Name{4B} & 0.79 & 0.88 & 0.98 & 0.97 & 0.78 & 0.90 & 0.97 & 1.00 \\
\hline
\end{tabular}
}\\
\subfloat[Matching time (ms)]{
\label{tab:scheme:3B:match}
\begin{tabular}{|l||rrrr|rrrr|}
\hline
 & \multicolumn{4}{c|}{\Name{JaWord}} & \multicolumn{4}{c|}{\Name{JaChars}} \\
Layout & 1K & 10K & 100K & 1M & 1K & 10K & 100K & 1M \\
\hline\hline
\Name{4B} & 310 & 402 & 484 & 731 & 304 & 439 & 620 & 1,403 \\
\Name{3B} & 334 & 430 & 526 & 783 & 320 & 451 & 618 & 1,456 \\
\hline
\Name{3B}/\Name{4B} & 1.08 & 1.07 & 1.09 & 1.07 & 1.06 & 1.03 & 1.00 & 1.04 \\
\hline
\end{tabular}
}\\
\subfloat[CPU instructions ($10^6$)]{
\label{tab:scheme:3B:instructions}
\begin{tabular}{|l||rrrr|rrrr|}
\hline
 & \multicolumn{4}{c|}{\Name{JaWord}} & \multicolumn{4}{c|}{\Name{JaChars}} \\
Layout & 1K & 10K & 100K & 1M & 1K & 10K & 100K & 1M \\
\hline\hline
\Name{4B} & 4,185 & 4,446 & 4,605 & 4,769 & 4,299 & 4,609 & 4,812 & 4,953 \\
\Name{3B} & 4,517 & 4,798 & 4,948 & 5,100 & 4,635 & 4,949 & 5,138 & 5,267 \\
\hline
\Name{3B}/\Name{4B} & 1.08 & 1.08 & 1.07 & 1.07 & 1.08 & 1.07 & 1.07 & 1.06 \\
\hline
\end{tabular}
}
\end{table}

Table \ref{tab:scheme:3B} shows the comparison results on \Name{JaWord} and \Name{JaChars}, where \Name{Mapped} has better performance than \Name{Bytewise}.
\Name{4B} and \Name{3B} indicate the cases that $\pi$ is implemented with four- and three-byte arrays, respectively.
As shown in Table \ref{tab:scheme:3B:memory}, the memory usage of \Name{3B} is necessarily smaller;
The smaller the number of patterns, the larger the difference.
This is because the size of $\pi$ is not linear to the number of patterns, although those of the other parts are linear.
As shown in Tables \ref{tab:scheme:3B:match}, however, the matching time of \Name{3B} is always slower.
Table \ref{tab:scheme:3B:instructions} shows the number of CPU instructions occurring during a pattern matching (measured with the \texttt{perf} command).
In \Name{3B}, converting three bytes into a four-byte integer in extraction of $\pi(c)$ requires more CPU instructions, resulting in slower matching.
These results suggest that when the dictionary scale is small, the 3B layout can reduce memory consumption while slightly sacrificing matching speed.

\subsection{Analysis on memory layouts}
\label{sect:exp:layout}

\begin{figure}[btp]
\centering
\subfloat[Matching time (ms)]{
\label{fig:layout:match:time}
\includegraphics[scale=0.5]{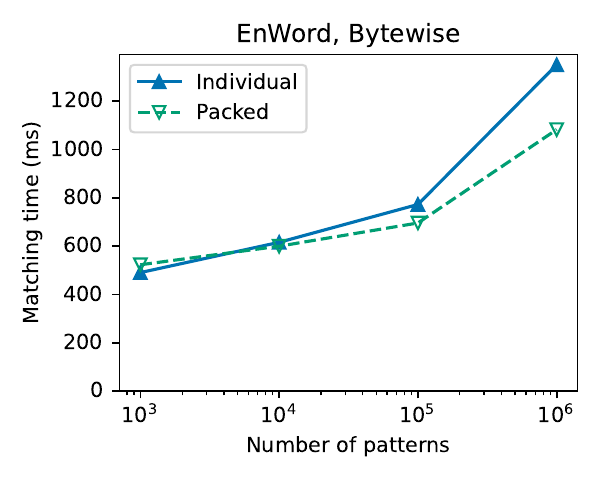}
\includegraphics[scale=0.5]{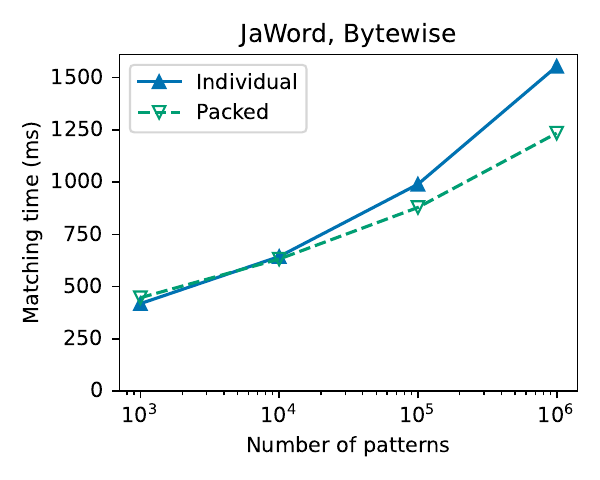}
\includegraphics[scale=0.5]{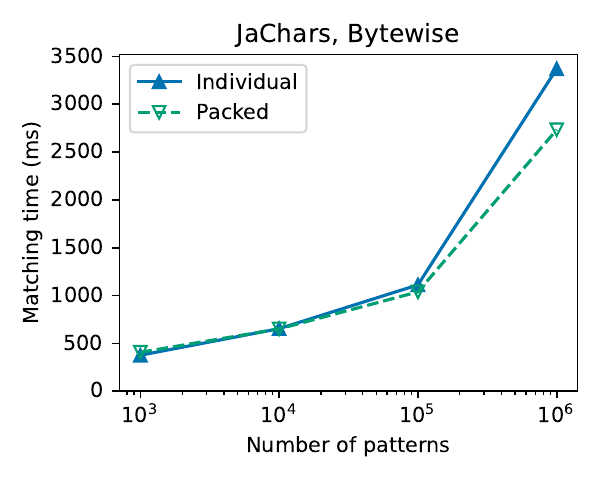}
}\\
\subfloat[Number of L1-load misses occurring during a pattern matching]{
\label{fig:layout:match:l1-cache}
\includegraphics[scale=0.5]{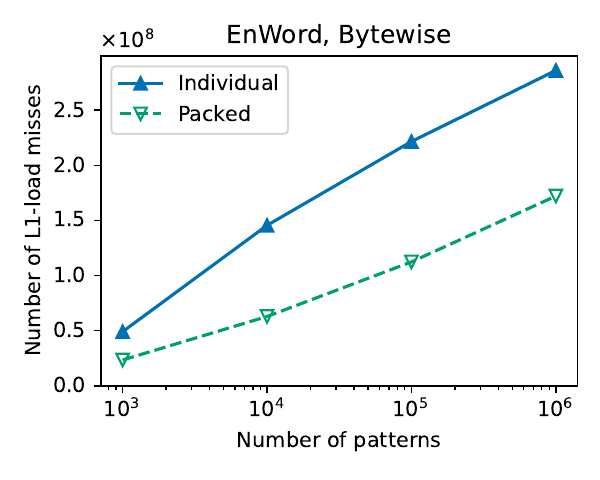}
\includegraphics[scale=0.5]{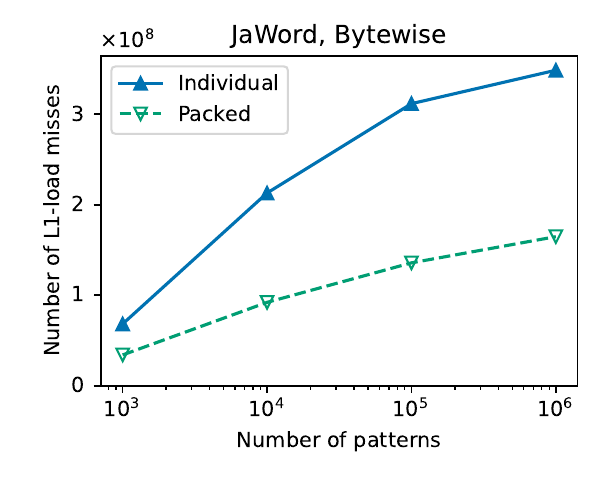}
\includegraphics[scale=0.5]{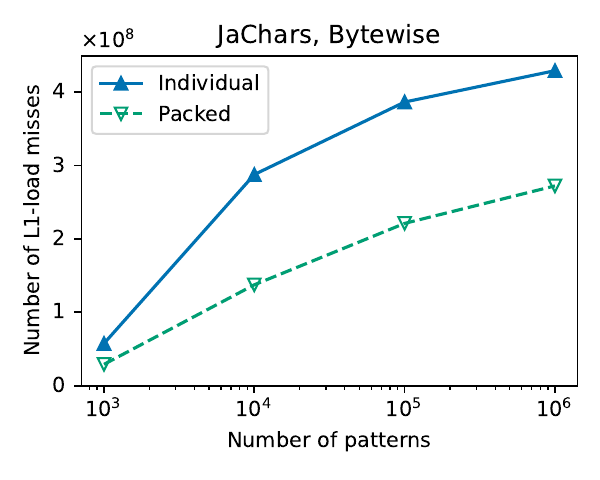}
}\\
\subfloat[Number of L3-load misses occurring during a pattern matching]{
\label{fig:layout:match:l3-cache}
\includegraphics[scale=0.5]{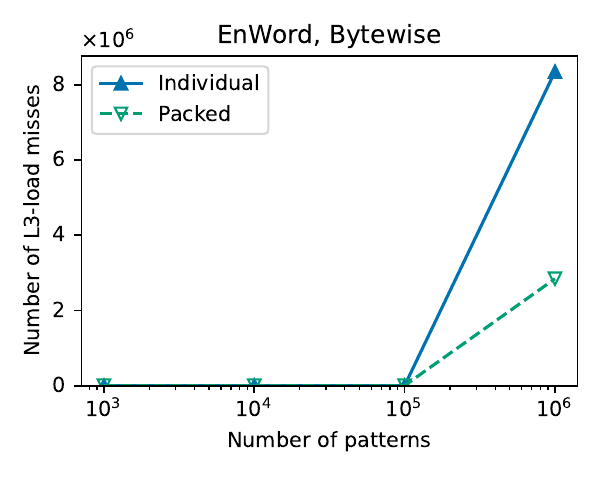}
\includegraphics[scale=0.5]{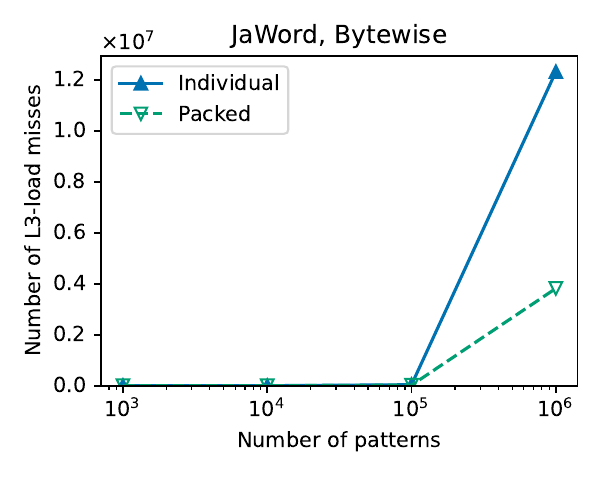}
\includegraphics[scale=0.5]{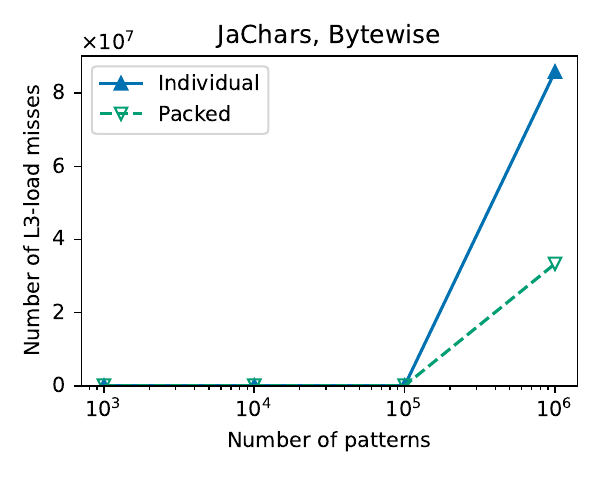}
}
\caption{
Comparison results on matching time and cache efficiency for the \Name{Individual} and \Name{Packed} layouts.
}
\label{fig:layout:match}
\end{figure}

Figure \ref{fig:layout:match} shows the experimental results for the \Name{Individual} and \Name{Packed} layouts presented in Section \ref{sect:tech:layout}, while fixing the other settings to \Name{Forest}, \Name{Bytewise}, \Name{Basic}, \Name{Chain}, and \Name{LexDFS}.
Figure \ref{fig:layout:match:time} shows the matching time.
The greater the number of patterns, the faster \Name{Packed} becomes.
\Name{Packed} is $\approx$20\% faster than \Name{Individual} when the number of patterns is $10^6$.

The time efficiency of \Name{Packed} stems from its cache-efficiency.
Figures \ref{fig:layout:match:l1-cache} and \ref{fig:layout:match:l3-cache} show the numbers of L1- and L3-load misses occurring during a pattern matching, respectively, measured using the \texttt{perf} command.
The number of L1-load misses in \Name{Packed} is always smaller than that in \Name{Individual}.
When the number of patterns is no greater than $10^5$, L3-load misses infrequently occur because the data structure fits in the cache memory.
However, when the number of patterns is $10^6$, \Name{Packed} suffered only 31--39\% of the L3-load misses experienced by \Name{Individual}.
From these results, we observe that the cache-efficiency with \Name{Packed} provides faster matching, especially for large dictionaries that cause costly L3-load misses.

\begin{figure}[btp]
\centering
\subfloat[Number of CPU instructions occurring during a pattern matching]{
\label{fig:layout:other:instructions}
\includegraphics[scale=0.5]{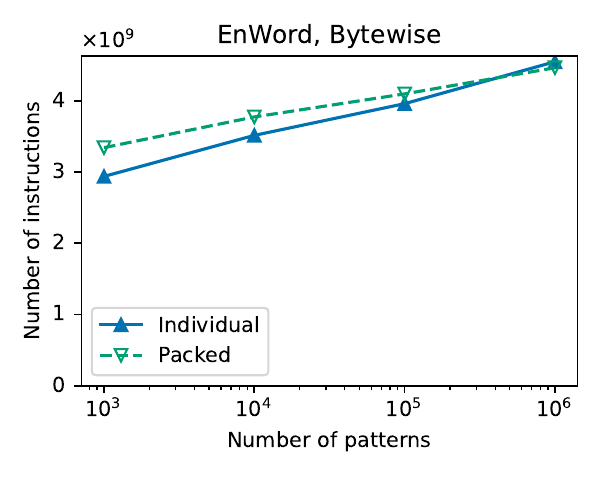}
\includegraphics[scale=0.5]{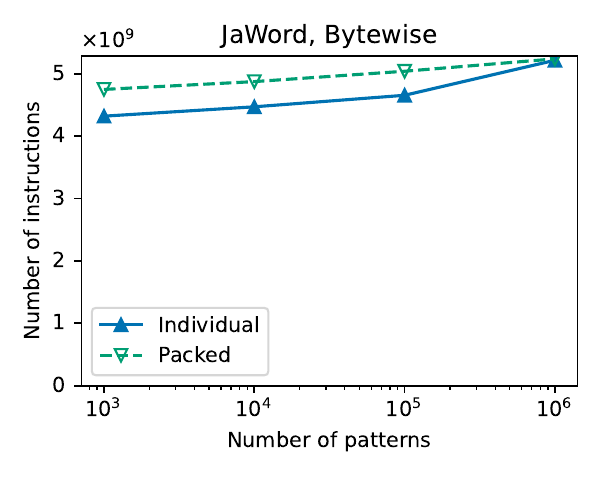}
\includegraphics[scale=0.5]{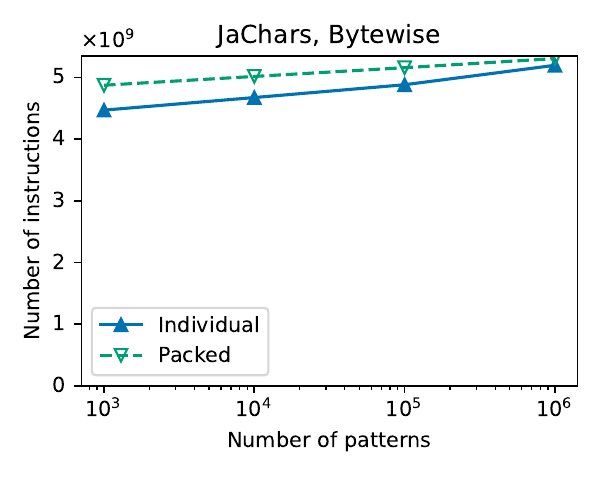}
}\\
\subfloat[Number of branch misses occurring during a pattern matching]{
\label{fig:layout:other:branch_miss}
\includegraphics[scale=0.5]{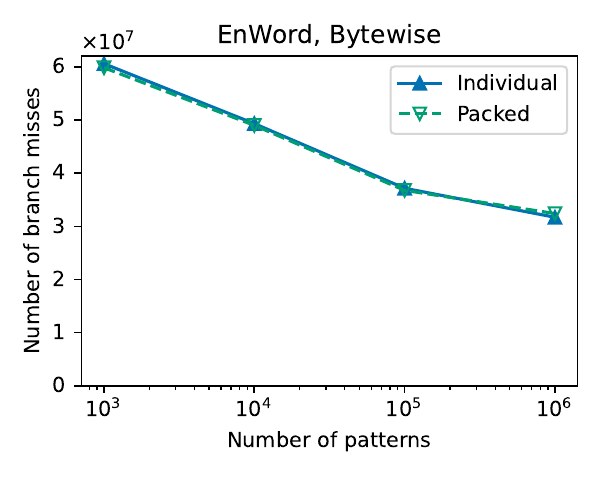}
\includegraphics[scale=0.5]{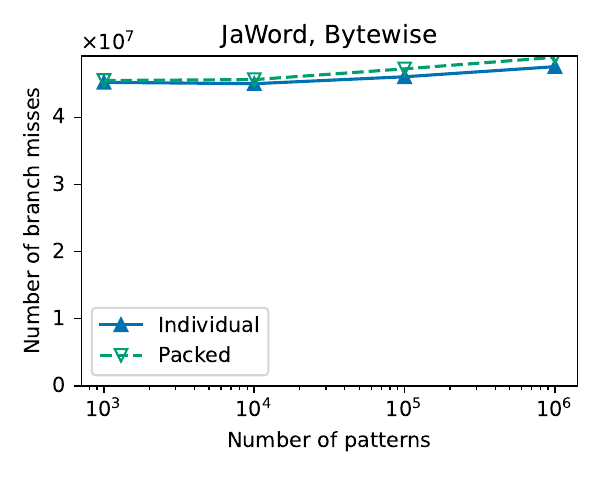}
\includegraphics[scale=0.5]{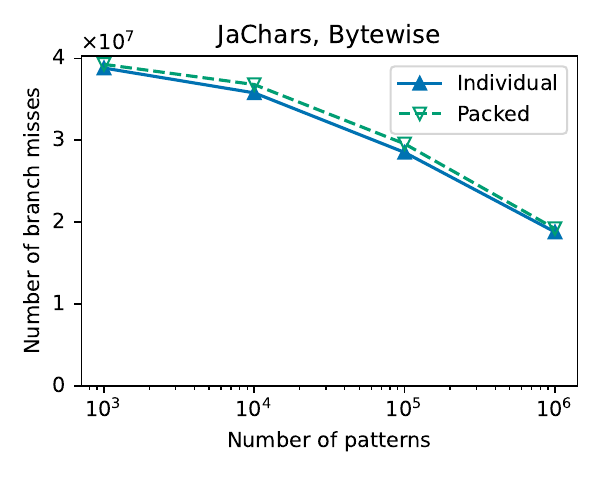}
}
\caption{
Comparison results on CPU instructions and branch misses for the \Name{Individual} and \Name{Packed} layouts.
}
\label{fig:layout:other}
\end{figure}

To reinforce our observation, we measured other metrics related to time performance.
Figures \ref{fig:layout:other:instructions} and \ref{fig:layout:other:branch_miss} show the numbers of CPU instructions and branch misses occurring during a pattern matching, respectively.
Since the only difference between \Name{Packed} and \Name{Individual} is the memory layout, these numbers do not differ significantly.
Rather, the number of CPU instructions in \Name{Packed} is slightly smaller (maybe because the data structure is implemented in a simpler manner without \texttt{struct}).
These results indicate that CPU instructions and branch misses are not really relevant to the comparison results in Figure \ref{fig:layout:match:time}.

\begin{table}[btp]
\footnotesize
\centering
\caption{
Proportion of the number of L3-load misses over that of L1 loads when the number of patterns is $10^6$.
}
\label{tab:layout:ratio}
\begin{tabular}{|l||rrr|}
\hline
Layout & \Name{EnWord} & \Name{JaWord} & \Name{JaChars} \\
\hline\hline
\Name{Individual} & 0.7\% & 1.0\% & 7.3\% \\
\Name{Packed} & 0.3\% & 0.4\% & 3.2\% \\
\hline
\end{tabular}
\end{table}

As a supplementary experiment, we also measured the proportion of the number of L3-load misses over that of L1 loads when the number of patterns is $10^6$.
Table \ref{tab:layout:ratio} shows the results.
The proportions in \Name{JaChars} are larger than those in \Name{EnWord} and \Name{JaWord}.
These results would be related to why matching in \Name{JaChars} takes more time in Figure \ref{fig:layout:match:time}:
\Name{Packed} takes 1.0--1.2 sec in \Name{EnWord} and \Name{JaWord} while it takes 2.7 sec in \Name{JaChars}, nevertheless CPU instructions and branch misses in \Name{JaChars} are not larger.
This suggests that it is important to further improve cache-efficiency to achieve faster matching, which would be helpful for future work.

\subsection{Analysis on array formats}
\label{sect:exp:format}

Table \ref{tab:format} shows the experimental results for the \Name{Basic} and \Name{Compact} formats presented in Section \ref{sect:tech:format}, while fixing the other settings to \Name{Forest}, \Name{Bytewise}, \Name{Packed}, \Name{Chain}, and \Name{LexDFS}.

\begin{table}[btp]
\footnotesize
\centering
\caption{
Experimental results for the \Name{Basic} and \Name{Compact} formats.
The memory usage is the total of \Name{BASE}, \Name{CHECK}, \Name{FAIL}, and \Name{OUTPOS}.
The last row in each table shows the different ratios for each result.
}
\label{tab:format}
\subfloat[Memory usage (KiB)]{
\label{tab:format:memory}
\begin{tabular}{|l||rrrr|rrrr|rrrr|}
\hline
 & \multicolumn{4}{c|}{\Name{EnWord}} & \multicolumn{4}{c|}{\Name{JaWord}} & \multicolumn{4}{c|}{\Name{JaChars}} \\
Format & 1K & 10K & 100K & 1M & 1K & 10K & 100K & 1M & 1K & 10K & 100K & 1M \\
\hline\hline
\Name{Basic} & 44.0 & 404 & 3,712 & 33,276 & 48.0 & 440 & 4,780 & 76,544 & 32.0 & 292 & 2,908 & 27,972 \\
\Name{Compact} & 33.0 & 303 & 2,784 & 24,957 & 36.0 & 330 & 3,585 & 57,408 & 24.0 & 219 & 2,181 & 20,979 \\
\hline
\Name{Compact}/\Name{Basic} & 0.75 & 0.75 & 0.75 & 0.75 & 0.75 & 0.75 & 0.75 & 0.75 & 0.75 & 0.75 & 0.75 & 0.75 \\
\hline
\end{tabular}
}\\
\subfloat[Matching time (ms)]{
\label{tab:format:match}
\begin{tabular}{|l||rrrr|rrrr|rrrr|}
\hline
 & \multicolumn{4}{c|}{\Name{EnWord}} & \multicolumn{4}{c|}{\Name{JaWord}} & \multicolumn{4}{c|}{\Name{JaChars}} \\
Format & 1K & 10K & 100K & 1M & 1K & 10K & 100K & 1M & 1K & 10K & 100K & 1M \\
\hline\hline
\Name{Basic} & 523 & 600 & 696 & 1,082 & 447 & 632 & 879 & 1,234 & 406 & 651 & 1,038 & 2,733 \\
\Name{Compact} & 545 & 606 & 712 & 1,052 & 473 & 633 & 893 & 1,226 & 415 & 609 & 1,028 & 2,439 \\
\hline
\Name{Compact}/\Name{Basic} & 1.04 & 1.01 & 1.02 & 0.97 & 1.06 & 1.00 & 1.02 & 0.99 & 1.02 & 0.94 & 0.99 & 0.89 \\
\hline
\end{tabular}
}\\
\subfloat[Construction time (ms)]{
\label{tab:format:build}
\begin{tabular}{|l||rrrr|rrrr|rrrr|}
\hline
 & \multicolumn{4}{c|}{\Name{EnWord}} & \multicolumn{4}{c|}{\Name{JaWord}} & \multicolumn{4}{c|}{\Name{JaChars}} \\
Format & 1K & 10K & 100K & 1M & 1K & 10K & 100K & 1M & 1K & 10K & 100K & 1M \\
\hline\hline
\Name{Basic} & 0.11 & 1.06 & 19.2 & 207 & 0.13 & 1.19 & 22.4 & 331 & 0.08 & 0.81 & 14.8 & 200 \\
\Name{Compact} & 0.14 & 1.30 & 23.5 & 249 & 0.14 & 1.39 & 32.7 & 182,015 & 0.09 & 0.91 & 16.0 & 226 \\
\hline
\Name{Compact}/\Name{Basic} & 1.25 & 1.23 & 1.22 & 1.20 & 1.09 & 1.17 & 1.46 & 549 & 1.09 & 1.12 & 1.08 & 1.13 \\
\hline
\end{tabular}
}
\end{table}

Table \ref{tab:format:memory} reports the total memory usage of \Name{BASE}, \Name{CHECK}, \Name{FAIL}, and \Name{OUTPOS}.
In all cases, the memory usage in \Name{Compact} is 75\% of that in \Name{Basic} because there is no difference in the resulting number of vacant ids (although Equation \eqref{eq:cda:base} has to be satisfied).
There is no significant difference in the matching time, as Table \ref{tab:format:match} shows.

Table \ref{tab:format:build} reports the construction time.
\Name{Compact} is always slower because of Equation \eqref{eq:cda:base},
especially on \Name{JaWord} of 1M patterns.
To clarify the slowdown, we investigated vacant ids for which vacant searches were mostly unsuccessful and found that the top-five vacant ids could only accept the five transition labels {0xC0}, {0xCF}, {0xD0}, {0xD1}, and {0xFF} to satisfy Equation \eqref{eq:cda:base}.
Because these labels did not appear or were very few in the dictionary, vacant searches for the vacant ids were unsuccessful many times, resulting in the slow construction.
Consequently, while the \Name{Compact} format does not produce a particularly high number of vacant ids, some of them can significantly slow down vacant searches when using \Name{Chain}.

\subsection{Analysis on acceleration techniques for vacant searches}
\label{sect:exp:constr}

\begin{table}[btp]
\footnotesize
\centering
\caption{
The average number of verifications per vacant search when using \Name{Chain}.
The number of patterns in each dictionary is one million.
}
\label{tab:constr:verif}
\begin{tabular}{|l||rrr|}
\hline
Method & \Name{EnWord} & \Name{JaWord} & \Name{JaChars} \\
\hline\hline
\Name{Bytewise} + \Name{Basic} & 2.8 & 2.8 & 6.6 \\
\Name{Bytewise} + \Name{Compact} & 4.6 & 5383.3 & 7.4 \\
\Name{Mapped} & 2.5 & 84.4 & 4934.5 \\
\hline
\end{tabular}
\end{table}

\begin{figure}[btp]
\centering
\subfloat[The average number of verifications per vacant search]{
\label{fig:exp:constr:verif}
\includegraphics[scale=0.6]{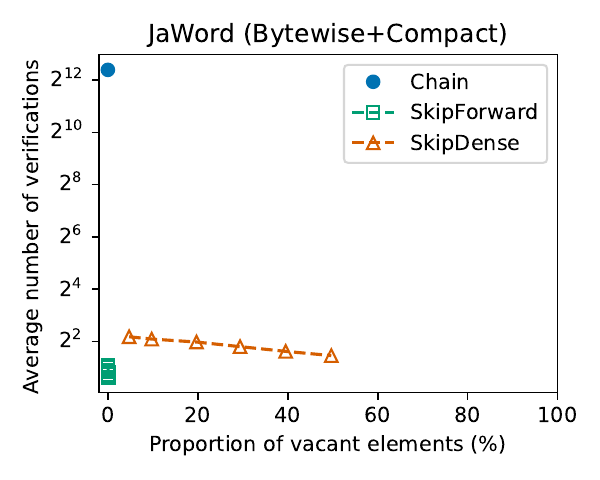}
\includegraphics[scale=0.6]{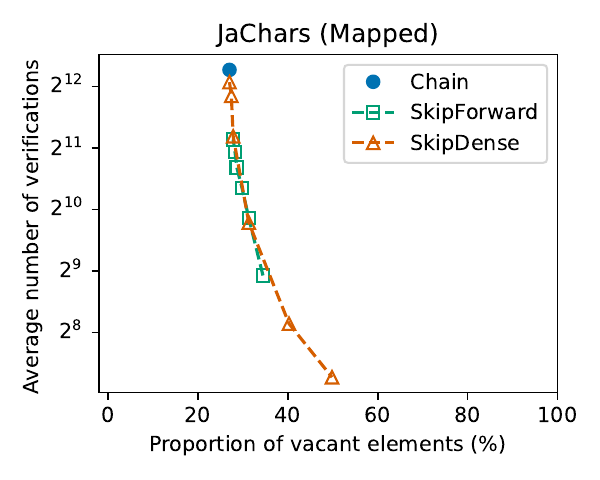}
}\\
\subfloat[Construction time in seconds]{
\label{fig:exp:constr:time}
\includegraphics[scale=0.6]{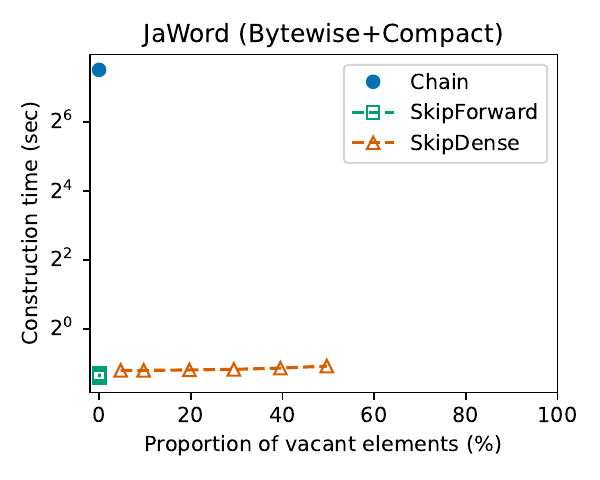}
\includegraphics[scale=0.6]{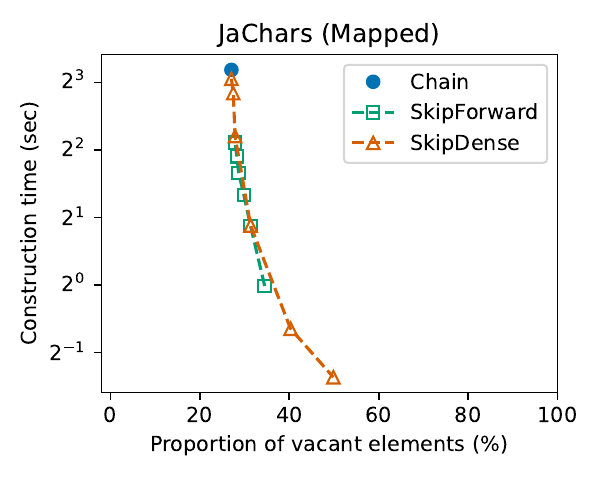}
}
\caption{
Experimental results for verifications and construction times when using \Name{Chain}, \Name{SkipForward}, or \Name{SkipDense}.
}
\label{fig:exp:constr}
\end{figure}

We first investigate problematic cases in DAAC construction when using \Name{Chain}.
Table \ref{tab:constr:verif} shows the average number of verifications per vacant search when using \Name{Chain}.
In the combination of \Name{Bytewise} and \Name{Basic}, the number is always several, indicating that it is hard to improve vacant searches further even using \Name{SkipForward} or \Name{SkipDense}.
However, the number is significantly higher in the cases of (i) \Name{Compact} on \Name{JaWord} and (ii) \Name{Mapped} on \Name{JaChars}.
Case (i) is because of Equation \eqref{eq:cda:base}, as discussed in Section \ref{sect:exp:format}.
Case (ii) is because of many vacant ids, as discussed in Section \ref{sect:exp:wise}.

We next test \Name{SkipForward} and \Name{SkipDense} to improve the two cases.
We test parameters $L = 4, 8, 12, 16, 20, 24$ in \Name{SkipForward} and $\tau = 0.05, 0.1, 0.2, 0.3, 0.4, 0.5$ in \Name{SkipDense}.
Figure \ref{fig:exp:constr:verif} shows the experimental results for the average number of verifications per vacant search.
The left figure corresponds to case (i), where we can see that
\Name{SkipForward} achieves several verifications with any parameter $L$ while maintaining similar vacant proportions.
\Name{SkipDense} also achieves several verifications, although the vacant proportions increase depending on the parameter $\tau$.
In the right figure corresponding to case (ii),
both \Name{SkipForward} and \Name{SkipDense} achieve fewer verifications than \Name{Chain}.
\Name{SkipForward} and \Name{SkipDense} plot similar curves, and their performances have no significant difference.

Figure \ref{fig:exp:constr:time} shows the experimental results for the construction times.
The plots are similar to those in Figure \ref{fig:exp:constr:verif}, and both \Name{SkipForward} and \Name{SkipDense} achieve construction times within several seconds.
These results indicate that we should choose \Name{SkipForward} or \Name{SkipDense} in accordance with the desired purpose:
if we want to control construction times, \Name{SkipForward} is best, and
if we want to control memory usage, \Name{SkipDense} is more suitable.

\subsection{Analysis on traversal orders}
\label{sect:exp:order}

\begin{figure}[btp]
\centering
\includegraphics[scale=0.5]{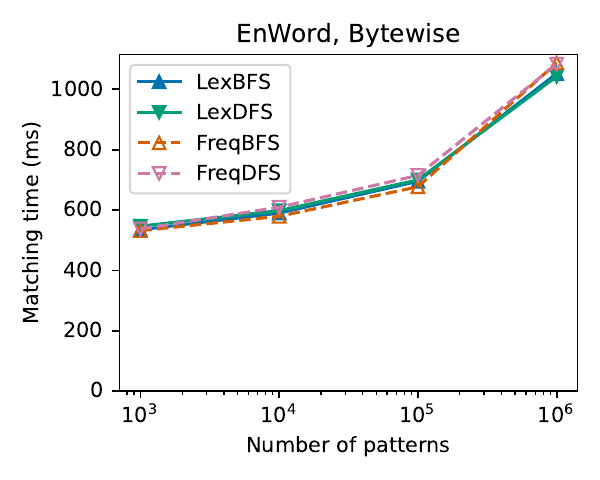}
\includegraphics[scale=0.5]{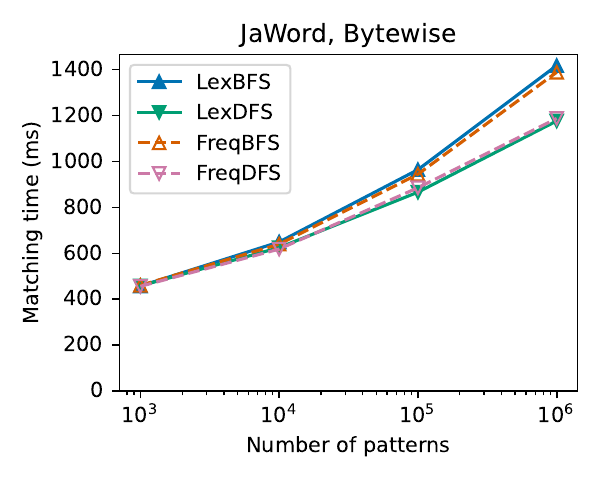}
\includegraphics[scale=0.5]{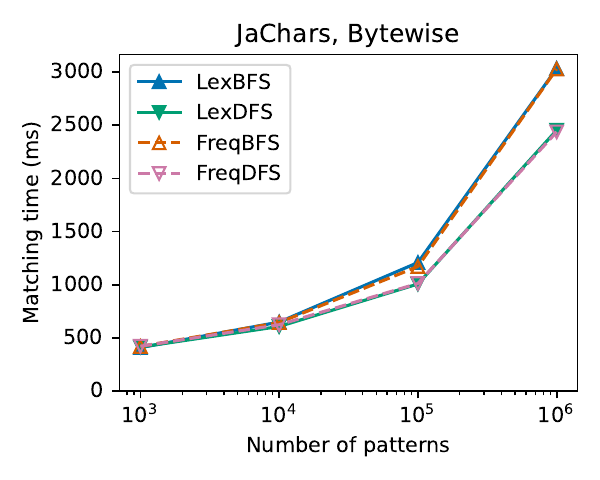}\\
\includegraphics[scale=0.5]{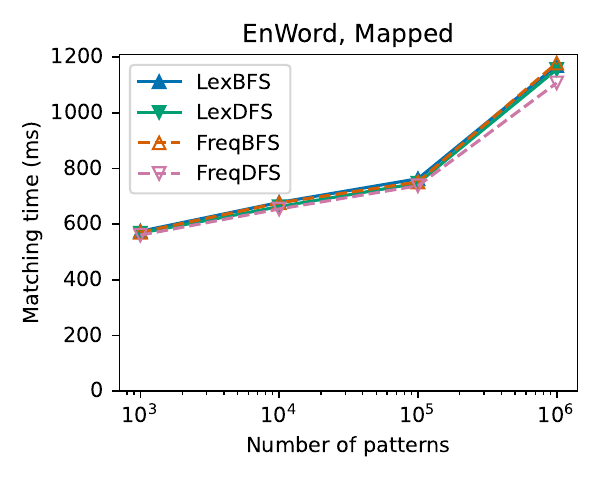}
\includegraphics[scale=0.5]{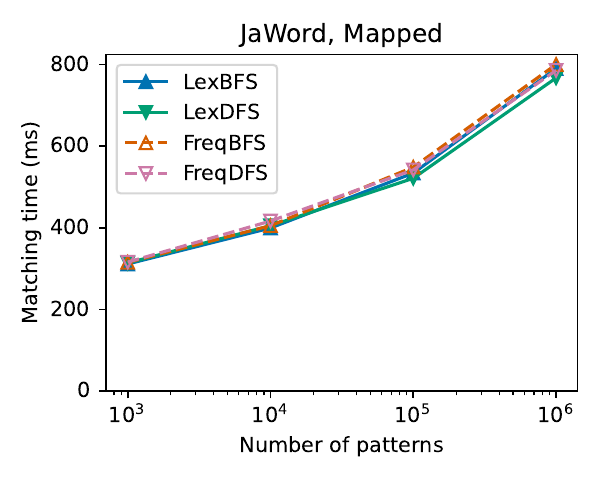}
\includegraphics[scale=0.5]{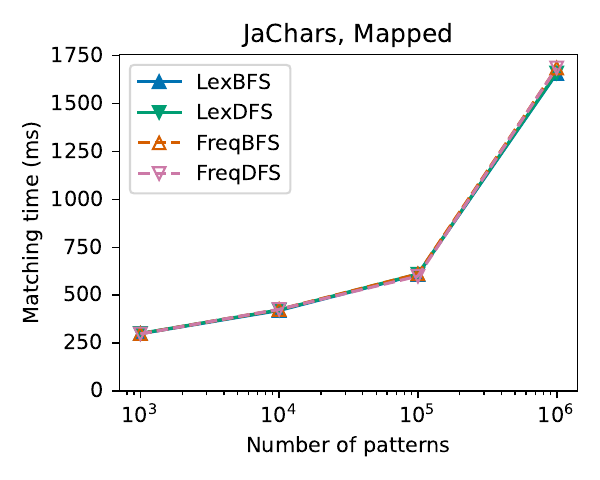}
\caption{
Matching times in milliseconds for different traversal orders.
}
\label{fig:order}
\end{figure}

Figure \ref{fig:order} shows the experimental results of matching times for \Name{LexBFS}, \Name{FreqBFS}, \Name{LexDFS}, and \Name{FreqDFS} presented in Section \ref{sect:tech:order}.
In the \Name{Bytewise} scheme, we fix the other settings to \Name{Forest}, \Name{Packed}, \Name{Compact}, and \Name{SkipForward} ($L=16$).
In the \Name{Mapped} scheme, we fix the other settings to \Name{Forest}, \Name{Packed}, and \Name{SkipForward} ($L=16$).

\begin{figure}[btp]
\centering
\includegraphics[scale=0.6]{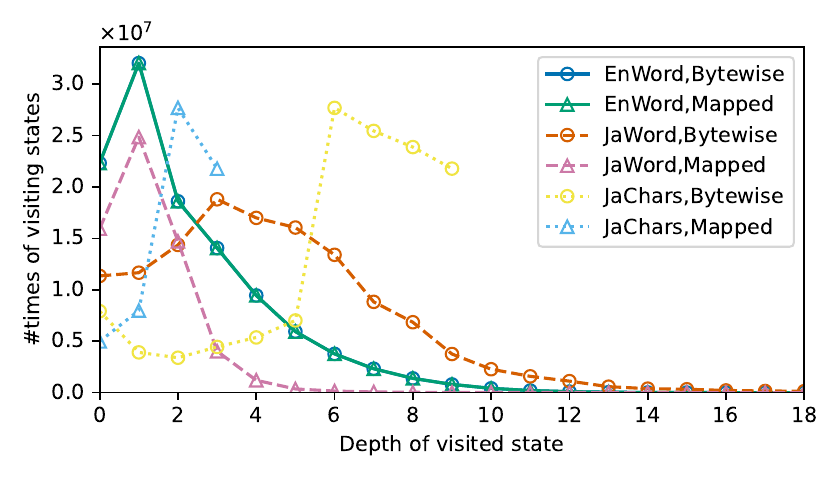}
\caption{
The number of visiting states at each depth during a pattern matching.
}
\label{fig:order:counts}
\end{figure}

\begin{table}[btp]
\footnotesize
\centering
\caption{
The number of L3-load misses occurring during a pattern matching with different traversal orders ({$\times 10^3$}).
The number of patterns is fixed to $10^6$.
The last row shows the different ratios for each result.
}
\label{tab:order:cache}
\begin{tabular}{|l||rr|rr|rr|}
\hline
 & \multicolumn{2}{c|}{\Name{EnWord}} & \multicolumn{2}{c|}{\Name{JaWord}} & \multicolumn{2}{c|}{\Name{JaChars}} \\
Order & \Name{Bytewise} & \Name{Mapped} & \Name{Bytewise} & \Name{Mapped} & \Name{Bytewise} & \Name{Mapped} \\
\hline\hline
\Name{LexBFS} & 2,485 & 2,967 & 4,301 & 3,157 & 31,999 & 17,719 \\
\Name{LexDFS} & 2,303 & 2,888 & 3,209 & 2,944 & 28,804 & 17,917 \\
\hline
\Name{LexDFS}/\Name{LexBFS} & 0.93 & 0.97 & 0.75 & 0.93 & 0.90 & 1.01 \\
\hline
\end{tabular}
\end{table}

We compare the \Name{DFS} and \Name{BFS} orders.
\Name{DFS} is faster than \Name{BFS} for \Name{JaWord} and \Name{JaChars} in \Name{Bytewise}.
To investigate the reason, we show the number of times states are visited in each depth during a pattern matching in Figure \ref{fig:order:counts}.
As we can see, deeper states are often visited for \Name{JaWord} and \Name{JaChars} in \Name{Bytewise}.
The number of L3-load misses occurring during a pattern matching is shown in Table \ref{tab:order:cache},
where we can see that \Name{DFS} is cache-efficient in most cases.
These observations indicate that \Name{DFS} enables cache-efficient traversals on deeper states and is suitable for long patterns.\footnote{As with Section \ref{sect:exp:layout}, we also measured the numbers of L1-load misses, CPU instructions, and branch misses; however, we do not show the numbers because they are not interesting as follows. On CPU instructions and branch misses, the difference between \Name{LexBFS} and \Name{LexDFS} was within $\pm 1.3\%$ and slight. This is essential because the traversal orders are related only to construction and the program in pattern matching is identical. Also on L1-load misses, the difference was within $\pm 10\%$ in most cases and does not impact the matching times reported in Figure \ref{fig:order}.}

Comparing \Name{Lex} and \Name{Freq}, there is no significant difference.
Although we also tested other orders, such as a random order, we did not observe any significant differences.
The reason is explained by the average number of outgoing transitions for an internal state reported in Table \ref{tab:scheme:outgoing}.
The number of outgoing transitions for an internal state was often small, indicating the order of visiting those outgoing transitions was not significant for achieving cache-efficient memory layout.

\subsection{Comparison with other AC automata}

On the basis of the above experimental analyses, we determine the best combinations of techniques and develop a new Rust library, called Daachorse, containing the following two variants:

\begin{itemize}
    \item \Name{Bytewise-Daachorse} = \Name{Forest} + \Name{Bytewise} + \Name{Packed} + \Name{Compact} + \Name{SkipForward} ($L=16$) + \Name{LexDFS}
    \item \Name{Charwise-Daachorse} = \Name{Forest} + \Name{Mapped} + \Name{Packed} + \Name{SkipForward} ($L=16$) + \Name{LexDFS}
\end{itemize}

We compare {Daachorse} with the \emph{Aho-corasick} library \cite{aho-corasick}, which is the most popular implementation of AC automata in Rust.
The data structure employs a hybrid form of the two pointer-based representations: the \emph{matrix and list forms} \cite{askitis2007efficient}.
The matrix form is fast but memory-inefficient, while the list form is memory-efficient but slow.
The hybrid form takes advantages of both by using them depending on the number of transitions for each state.
This library provides two types of implementations:
\Name{NFA-AC} is a standard AC automaton (as described in Section \ref{sect:pre:ac}), and
\Name{DFA-AC} is a deterministic finite automaton compiled from \Name{NFA-AC}.
\Name{DFA-AC} is an option for faster matching, although it consumes a huge amount of memory \cite[Chapter 3.2.3]{navarro2002flexible}.

\begin{figure}[tbp]
\centering
\subfloat[Matching time (ms)]{
\includegraphics[scale=0.5]{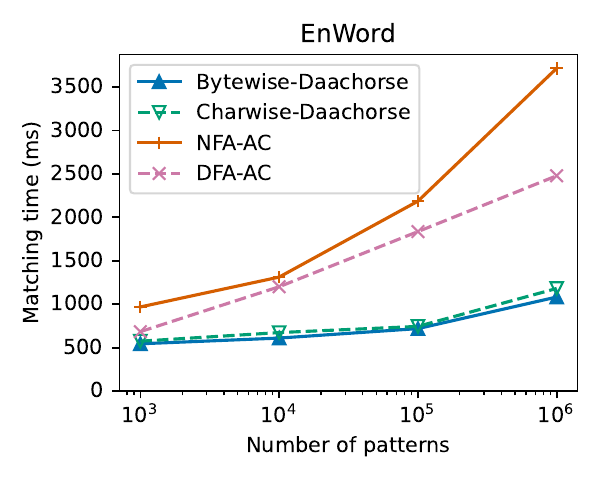}
\includegraphics[scale=0.5]{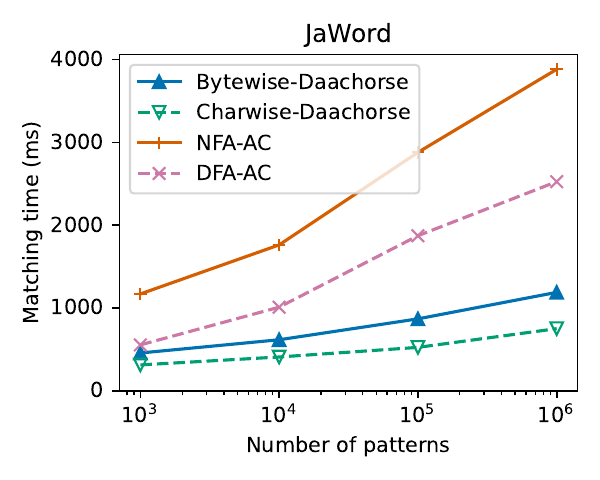}
\includegraphics[scale=0.5]{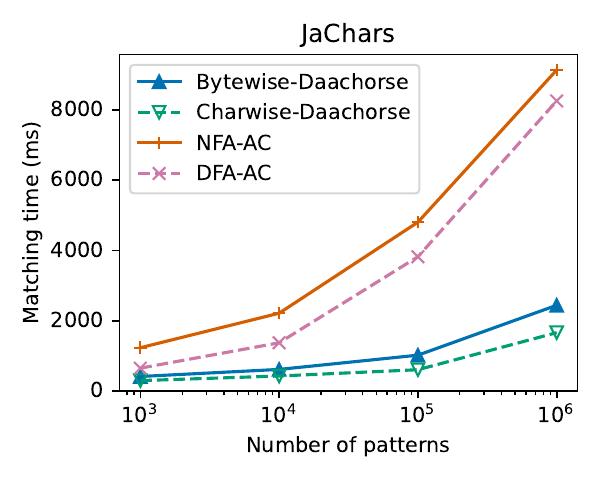}
\label{fig:comparison:match}
}\\
\subfloat[Construction time (ms)]{
\includegraphics[scale=0.5]{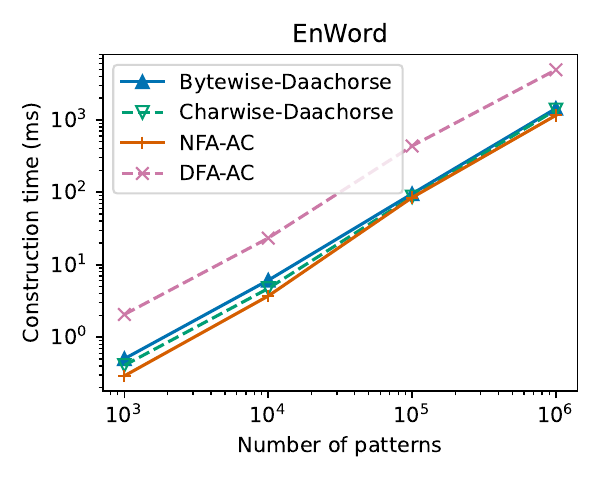}
\includegraphics[scale=0.5]{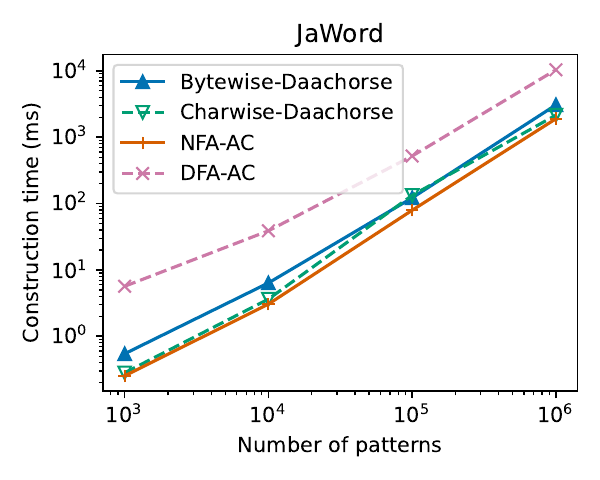}
\includegraphics[scale=0.5]{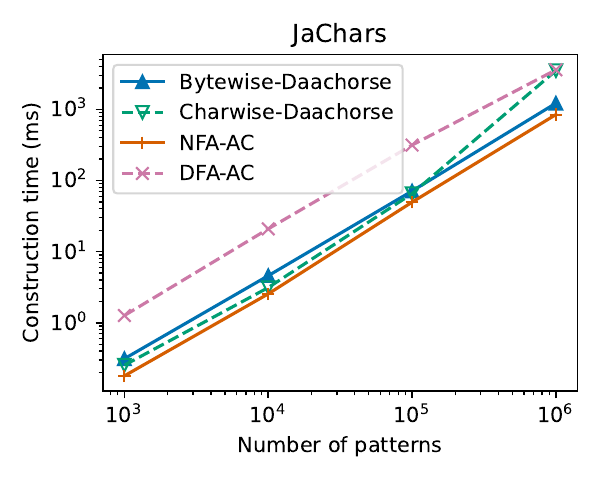}
\label{fig:comparison:constr}
}\\
\subfloat[Memory usage (KiB)]{
\includegraphics[scale=0.5]{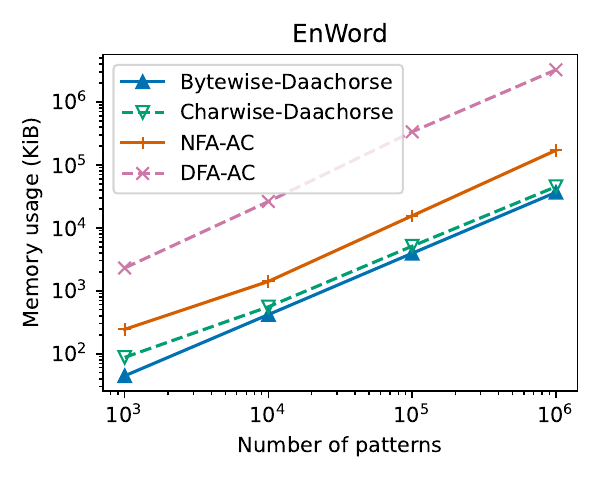}
\includegraphics[scale=0.5]{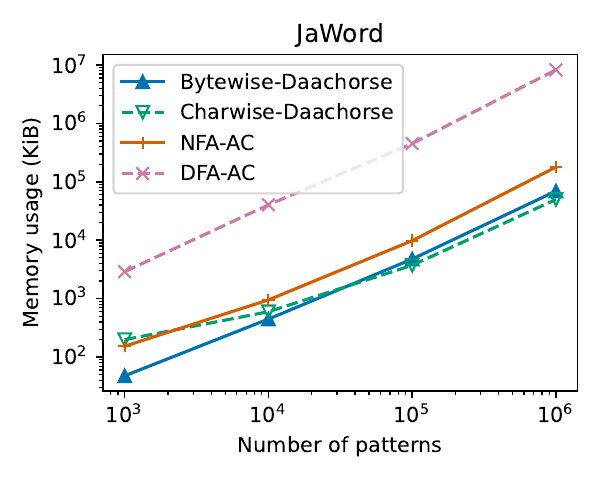}
\includegraphics[scale=0.5]{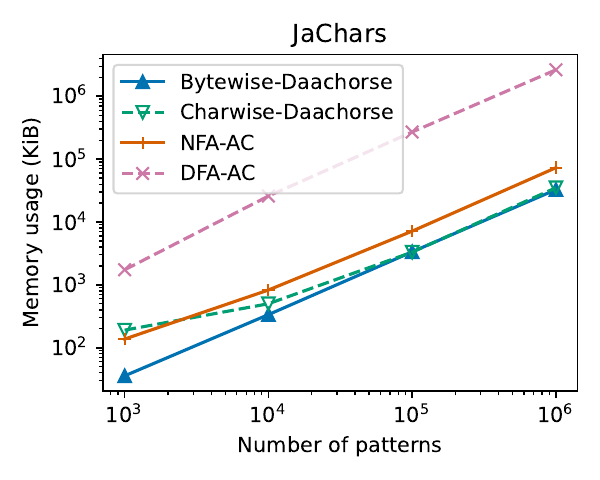}
\label{fig:comparison:size}
}
\caption{
Comparison results with the Daachorse and Aho-corasick libraries.
(a) The matching time is plotted in a linear scale because the time complexity is not linear over the number of patterns.
(b) The construction time and (c) the memory usage are plotted in a logarithmic scale because the complexities are linear over the number of patterns.
}
\label{fig:comparison}
\end{figure}

Figure \ref{fig:comparison} shows the comparison results.
For the matching times reported in Figure \ref{fig:comparison:match}, \Name{Bytewise-Daachorse} is always the fastest on \Name{EnWord}, and \Name{Charwise-Daachorse} is always the fastest on \Name{JaWord} and \Name{JaChars}.
Although \Name{DFA-AC} is also fast when the number of patterns is $10^3$, its performance degrades as the number of patterns increases.
For the construction times reported in Figure \ref{fig:comparison:constr}, \Name{NFA-AC} is always the fastest, and \Name{Bytewise} and \Name{Charwise-Daachorse} approach \Name{NFA-AC}.
\Name{DFA-AC} is always the slowest because compiling from \Name{NFA-AC} takes time.
For the memory usages reported in Figure \ref{fig:comparison:size}, \Name{Bytewise-} and \Name{Charwise-Daachorse} are often the smallest;
however, when the number of patterns is $10^3$, \Name{Charwise-Daachorse} is larger than \Name{NFA-AC} because of the memory consumption of the mapping $\pi$.

\section{Application example: Japanese tokenization}
\label{sect:syst}

Vaporetto \cite{vaporetto}, a supervised Japanese tokenizer written in Rust, is an application example of AC automata.
Vaporetto is an efficient implementation of the pointwise prediction algorithm \cite{mori2011pointwise} and leverages the AC algorithm in its feature extraction phase.
To demonstrate the ability of Daachorse, we integrated the Daachorse and Aho-corasick libraries with Vaporetto and evaluated the tokenization speeds.

We conducted these experiments in the same environment as the experiments in Section \ref{sect:exp}.
To train the Vaporetto model, we used 60K sentences with manually annotated short-unit-word boundaries in BCCWJ (version 1.1) \cite{maekawa2014bccwj} and 667K unique words in UniDic (version 3.1.0) \cite{den2007unidic}, while specifying default parameters in Vaporetto.
We performed tokenization for 5.9M sentences in BCCWJ that are not used for the training and then measured the running time.

The experimental result showed that:

\begin{itemize}
\item \Name{Bytewise-Daachorse} performed in 3.3 microseconds per sentence,
\item \Name{Charwise-Daachorse} performed in 2.9 microseconds per sentence,
\item \Name{NFA-AC} performed in 7.5 microseconds per sentence, and
\item \Name{DFA-AC} performed in 7.3 microseconds per sentence.
\end{itemize}

Daachorse was at most \Times{2.6} faster than the Aho-corasick library.
This result indicates that the time efficiency of the AC algorithm is critical in Vaporetto's tokenization, demonstrating the utility of our Daachorse.

\section{Related works}
\label{sect:related}

\paragraph{Double-array}

The two most popular implementations of DAACs are
the Java library by hankcs \cite{AhoCorasickDoubleArrayTrie} and the Go library by Ruohang \cite{Vonng-ac}.
Both libraries are implemented in a similar manner and employ the same techniques: \Name{Simple} (although the memory layout is not identical), \Name{Charwise}, \Name{Individual}, \Name{SkipDense} ($\tau=0.05$, without \Name{Chain} and the block partitioning), and \Name{LexBFS}.
As demonstrated in Section \ref{sect:exp}, these techniques do not provide the best performances in most cases, and the implementations could be improved by applying \Name{Forest} instead of \Name{Simple}, \Name{Mapped} instead of \Name{Charwise}, \Name{Packed} instead of \Name{Individual}, \Name{SkipForward} instead of \Name{SkipDense}, and \Name{LexDFS} instead of \Name{LexBFS}.

Some studies have represented the integers of \Name{BASE} and \Name{CHECK} in a compressed space other than \Name{Compact} (described in Section \ref{sect:tech:format}).
Fuketa et al. \cite{fuketa2004new} proposed partitioning one double-array structure into several smaller ones to implement \Name{BASE} and \Name{CHECK} consisting of 2-byte integers;
however, as this approach produces many structures for a large trie, many pointers must be maintained to connect the structures.
Fuketa et al. \cite{fuketa2014compression} proposed eliminating \Name{BASE} by introducing additional character mappings, although its applications are limited to fixed-length keywords such as zip codes.
Kanda et al. \cite{kanda2016compression,kanda2017compressed} proposed approaches to compress the integers of \Name{BASE} and \Name{CHECK} through differential encoding;
however, time efficiency can degrade because of more computations.

Another approach to compress the double-array is to employ compact trie forms.
The \emph{minimal-prefix (MP) trie} \cite{aoe1992efficient,dundas1991implementing} is an often-used form \cite{aoe1989efficient,kanda2017compressed,yata2007compact}, and various methods to efficiently implement the MP-trie using the double-array have been proposed \cite{yata2007efficient,dorji2010new,kanda2017rearrangement,yata2006double}.
However, the MP-trie is specialized for dictionary lookups, which are done by traversing a trie from the root to a leaf.
Other compact forms such as double-tries \cite{aoe1996trie} and directed acyclic word graphs \cite{yata2008fast} cannot be applied to AC automata in principle.

Liu et al. \cite{liu2011compression} showed that double-array structures resulting from large alphabets (e.g., multibyte characters) can include many vacant ids and are thus memory inefficient.
They proposed several approaches to address this problem and demonstrated that code mapping with the frequency of characters is effective on Chinese strings.
The efficiency of the code mapping was also demonstrated on $N$-gram sequences \cite{norimatsu2016fast,yasuhara2013efficient}.

\paragraph{Alternative trie representations}

Tries \cite{fredkin1960trie} have been studied since the 1960s, and there are many data structures to represent.
Since the AC automaton is a simple extension of the trie, we have many alternatives to the double-array.
The conventional data structures are a \emph{matrix form} and a \emph{list form} \cite{askitis2007efficient}.
The matrix form uses a transition matrix of size $|S| \times |\Sigma|$ and performs a transition lookup in $O(1)$ time, while consuming a large space of $O(|S||\Sigma|)$.
The list form stores a set of transition labels from each state in a sorted array and performs a transition lookup by binary search in $O(\log |\Sigma|)$ time, while consuming a smaller space of $O(|S|)$.
The matrix and list forms have a time-space trade-off, and the double-array harnesses both advantages by compressing the matrix form.

There are other data structures that utilize both advantages of the matrix and list forms.
One is a {hybrid} approach that uses the matrix form for states with many outgoing transitions and the list form for states with few outgoing transitions.
The aho-corasick library \cite{aho-corasick} employs the hybrid approach and is outperformed by Daachorse, as we demonstrated.
Another approach is {hashing} \cite{askitis2007efficient}, which stores mappings from states to outgoing transitions in hash tables.
This approach can perform a transition lookup in $O(1)$ expected time while consuming $O(|S|)$ space.
However, searching in the hash table often requires more computations than transition lookups in the double-array.
Prior experiments \cite{nieminen2007efficient} demonstrated that DAACs outperformed AC automata with the hashing approach.

Compressed representations of tries have recently been proposed to store massive datasets in main memory.
\emph{Succinct tries} \cite{arroyuelo2010succinct} are representative data structures.
Their memory usage achieves $|S| \log_2 |\Sigma| + O(|S|)$ bits of space and is close to the information-theoretic lower bound \cite{jacobson1989space}.
However, the succinct tries employ many bit manipulations in tree navigational operations, and their time efficiency is not competitive with that of double-arrays \cite{kanda2016compression,kanda2017compressed}.
More compressed data structures, such as XBW \cite{ferragina2005structuring} and Elias-Fano tries \cite{pibiri2017efficient,pibiri2019handling}, also have similar time bottlenecks and are not suited to design fast AC automata.

\paragraph{Compressed representations of AC automata}

Another line of research proposes data structures to represent AC automata (i.e., transition, failure, and output functions) in a compressed space.
Belazzougui \cite{belazzougui2010succinct} proposed the first compressed data structure that supports a matching in optimal $O(n + \textsf{occ})$ time, while achieving $|S| \log_2 |\Sigma| + O(|S|)$ bits of space.
Hon et al. \cite{hon2013faster} achieved an entropy compressed space while matching time remains optimal.
I et al. \cite{tomohiro2015compressed} designed a matching algorithm working on grammar-based compressed AC automata.
However, these studies were accomplished through theoretical discussions, and we are unaware of any actual implementation.

\section{Conclusion}
\label{sect:conc}

In this paper, we provided a comprehensive description of implementation techniques in DAACs and experimentally revealed the most efficient combinations of the techniques to achieve higher performance.
We also designed a data structure of DAACs and developed a new Rust library for faster multiple pattern matching, called Daachorse.
Our experiments showed that, compared to other implementations of AC automata, Daachorse offered a superior performance in terms of time and space efficiency.
As we demonstrated in the integration test with Vaporetto, Daachorse has significant potential to improve applications that employ multiple pattern matching.

Our future work will incorporate Daachorse in other applications to enable faster text processing.
We are also interested in the performance changes of data structures in different environments such as small devices.
Another future work will conduct high performance tuning under various environments and propose data structures optimized for those environments.

\section*{Acknowledgement}

We thank Shinsuke Mori for his cooperation and Keisuke Goto for his insightful review.

\bibliographystyle{plain}
\bibliography{library}

\begin{thebibliography}{10}

\bibitem{aho1975efficient}
Alfred~V Aho and Margaret~J Corasick.
\newblock {Efficient string matching: An aid to bibliographic search}.
\newblock {\em Communications of the ACM}, 18(6):333--340, 1975.

\bibitem{aoe1989efficient}
Jun'ichi Aoe.
\newblock {An efficient digital search algorithm by using a double-array
  structure}.
\newblock {\em IEEE Transactions on Software Engineering}, 15(9):1066--1077,
  1989.

\bibitem{aoe1992efficient}
Jun'ichi Aoe, Katsushi Morimoto, and Takashi Sato.
\newblock {An efficient implementation of trie structures}.
\newblock {\em Software: Practice and Experience}, 22(9):695--721, 1992.

\bibitem{aoe1996trie}
Jun'ichi Aoe, Katsushi Morimoto, Masami Shishibori, and Ki-Hong Park.
\newblock {A trie compaction algorithm for a large set of keys}.
\newblock {\em IEEE Transactions on Knowledge and Data Engineering},
  8(3):476--491, 1996.

\bibitem{arroyuelo2010succinct}
Diego Arroyuelo, Rodrigo C{\'a}novas, Gonzalo Navarro, and Kunihiko Sadakane.
\newblock Succinct trees in practice.
\newblock In {\em Proceedings of the 12th Workshop on Algorithm Engineering and
  Experiments (ALENEX)}, pages 84--97, 2010.

\bibitem{askitis2007efficient}
Nikolas Askitis.
\newblock {\em {Efficient data structures for cache architectures}}.
\newblock PhD thesis, RMIT University, 2007.

\bibitem{belazzougui2010succinct}
Djamal Belazzougui.
\newblock Succinct dictionary matching with no slowdown.
\newblock In {\em Proceedings of the 21st Annual Symposium on Combinatorial
  Pattern Matching (CPM)}, pages 88--100, 2010.

\bibitem{brants2006web}
Thorsten Brants and Alex Franz.
\newblock {Web 1T 5-gram Version 1}.
\newblock {\em Linguistic Data Consortium}, 2006.

\bibitem{groonga}
{Brazil Inc}.
\newblock {Groonga: An open-source fulltext search engine and column store
  (12.0.3)}.
\newblock \url{https://groonga.org/}, 2022.

\bibitem{vaporetto}
{daac-tools}.
\newblock {Vaporetto: Very Accelerated POintwise pREdicTion based TOkenizer
  (0.4.0)}.
\newblock \url{https://github.com/daac-tools/vaporetto}, 2022.

\bibitem{den2007unidic}
Yasuharu Den, Toshinobu Ogiso, Hideki Ogura, Atsushi Yamada, Nobuaki Minematsu,
  Kiyotaka Uchimoto, and Hanae Koiso.
\newblock The development of an electronic dictionary for morphological
  analysis and its application to {J}apanese corpus linguistics (in japanese).
\newblock {\em Japanese linguistics}, 22:101--123, 2007.

\bibitem{dorji2010new}
Tshering~C Dorji, El-sayed Atlam, Susumu Yata, Mahmoud Rokaya, Masao Fuketa,
  Kazuhiro Morita, and Jun'ichi Aoe.
\newblock {New methods for compression of MP double array by compact management
  of suffixes}.
\newblock {\em Information Processing {\&} Management}, 46(5):502--513, 2010.

\bibitem{dundas1991implementing}
John~A Dundas.
\newblock {Implementing dynamic minimal-prefix tries}.
\newblock {\em Software: Practice and Experience}, 21(10):1027--1040, 1991.

\bibitem{ferragina2005structuring}
Paolo Ferragina, Fabrizio Luccio, Giovanni Manzini, and S~Muthukrishnan.
\newblock {Structuring labeled trees for optimal succinctness, and beyond}.
\newblock In {\em Proceedings of the 46th Annual IEEE Symposium on Foundations
  of Computer Science (FOCS)}, pages 184--193, 2005.

\bibitem{pctext}
Paolo Ferragina and Gonzalo Navarro.
\newblock {Pizza\&Chili Corpus}.
\newblock \url{http://pizzachili.dcc.uchile.cl/texts.html}, 2005.

\bibitem{fredkin1960trie}
Edward Fredkin.
\newblock {Trie memory}.
\newblock {\em Communications of the ACM}, 3(9):490--499, 1960.

\bibitem{fuketa2014compression}
Masao Fuketa, Hiroya Kitagawa, Takuki Ogawa, Kazuhiro Morita, and Jun-ichi Aoe.
\newblock Compression of double array structures for fixed length keywords.
\newblock {\em Information processing \& management}, 50(5):796--806, 2014.

\bibitem{fuketa2004new}
Masao Fuketa, Kazuhiro Morita, Toru Sumitomo, Shinkaku Kashiji, Elsayed Atlam,
  and Jun-Ichi Aoe.
\newblock A new compression method of double array for compact dictionaries.
\newblock {\em International Journal of Computer Mathematics}, 81(8):943--953,
  2004.

\bibitem{aho-corasick}
Andrew Gallant.
\newblock {aho-corasick: A fast implementation of Aho-Corasick in Rust
  (0.7.18)}.
\newblock \url{https://github.com/BurntSushi/aho-corasick}, 2021.

\bibitem{AhoCorasickDoubleArrayTrie}
hankcs.
\newblock {AhoCorasickDoubleArrayTrie: An extremely fast implementation of Aho
  Corasick algorithm based on Double Array Trie (1.2.2)}.
\newblock \url{https://github.com/hankcs/AhoCorasickDoubleArrayTrie}, 2020.

\bibitem{hon2013faster}
Wing-Kai Hon, Tsung-Han Ku, Rahul Shah, Sharma~V Thankachan, and Jeffrey~Scott
  Vitter.
\newblock Faster compressed dictionary matching.
\newblock {\em Theoretical Computer Science}, 475:113--119, 2013.

\bibitem{tomohiro2015compressed}
Tomohiro I, Takaaki Nishimoto, Shunsuke Inenaga, Hideo Bannai, and Masayuki
  Takeda.
\newblock Compressed automata for dictionary matching.
\newblock {\em Theoretical Computer Science}, 578:30--41, 2015.

\bibitem{jacobson1989space}
Guy Jacobson.
\newblock {Space-efficient static trees and graphs}.
\newblock In {\em Proceedings of the 30th IEEE Symposium on Foundations of
  Computer Science (FOCS)}, pages 549--554, 1989.

\bibitem{kanda2017rearrangement}
Shunsuke Kanda, Yuma Fujita, Kazuhiro Morita, and Masao Fuketa.
\newblock {Practical rearrangement methods for dynamic double-array
  dictionaries}.
\newblock {\em Software: Practice and Experience}, 48(1):65--83, 2018.

\bibitem{kanda2016compression}
Shunsuke Kanda, Masao Fuketa, Kazuhiro Morita, and Jun'ichi Aoe.
\newblock {A compression method of double-array structures using linear
  functions}.
\newblock {\em Knowledge and Information Systems}, 48(1):55--80, 2016.

\bibitem{kanda2017compressed}
Shunsuke Kanda, Kazuhiro Morita, and Masao Fuketa.
\newblock Compressed double-array tries for string dictionaries supporting fast
  lookup.
\newblock {\em Knowledge and Information Systems}, 51(3):1023--1042, 2017.

\bibitem{kudo-etal-2004-applying}
Taku Kudo, Kaoru Yamamoto, and Yuji Matsumoto.
\newblock Applying conditional random fields to {J}apanese morphological
  analysis.
\newblock In {\em Proceedings of the 2004 Conference on Empirical Methods in
  Natural Language Processing (EMNLP)}, pages 230--237, 2004.

\bibitem{darts}
Takuu Kudo.
\newblock {Darts: Double-ARray Trie System (0.32)}.
\newblock \url{http://chasen.org/~taku/software/darts/}, 2008.

\bibitem{liu2011compression}
Huidan Liu, Minghua Nuo, Long-Long Ma, Jian Wu, and Yeping He.
\newblock {Compression methods by code mapping and code dividing for chinese
  dictionary stored in a double-array trie}.
\newblock In {\em Proceedings of the 5th International Joint Conference on
  Natural Language Processing (IJCNLP)}, pages 1189--1197, 2011.

\bibitem{maekawa2014bccwj}
Kikuo Maekawa, Makoto Yamazaki, Toshinobu Ogiso, Takehiko Maruyama, Hideki
  Ogura, Wakako Kashino, Hanae Koiso, Masaya Yamaguchi, Makiro Tanaka, and
  Yasuharu Den.
\newblock Balanced corpus of contemporary written {J}apanese.
\newblock {\em Language resources and evaluation}, 48(2):345–371, jun 2014.

\bibitem{mori2011pointwise}
Shinsuke Mori, Yosuke Nakata, Graham Neubig, and Tatsuya Kawahara.
\newblock {M}orphological {A}nalysis with {P}ointwise {P}redictors (in
  {J}apanese).
\newblock {\em Journal of Natural Language Processing}, 18(4):367--381, 2011.

\bibitem{morita2001fast}
Kazuhiro Morita, Masao Fuketa, Yoshihiro Yamakawa, and Jun'ichi Aoe.
\newblock {Fast insertion methods of a double-array structure}.
\newblock {\em Software: Practice and Experience}, 31(1):43--65, 2001.

\bibitem{navarro2002flexible}
Gonzalo Navarro and Mathieu Raffinot.
\newblock {\em {Flexible pattern matching in strings: Practical on-line search
  algorithms for texts and biological sequences}}.
\newblock Cambridge university press, 2002.

\bibitem{neubig2011pointwise}
Graham Neubig, Yosuke Nakata, and Shinsuke Mori.
\newblock Pointwise prediction for robust, adaptable {J}apanese morphological
  analysis.
\newblock In {\em Proceedings of the 49th Annual Meeting of the Association for
  Computational Linguistics: Human Language Technologies}, pages 529--533,
  2011.

\bibitem{nieminen2007efficient}
Janne Nieminen and Pekka Kilpel{\"a}inen.
\newblock {Efficient implementation of Aho--Corasick pattern matching automata
  using unicode}.
\newblock {\em Software: Practice and Experience}, 37(6):669--690, 2007.

\bibitem{norimatsu2016fast}
Jun'ya Norimatsu, Makoto Yasuhara, Toru Tanaka, and Mikio Yamamoto.
\newblock {A fast and compact language model implementation using double-array
  structures}.
\newblock {\em ACM Transactions on Asian and Low-Resource Language Information
  Processing}, 15(4):27, 2016.

\bibitem{oono2003fast}
Masaki Oono, El-Sayed Atlam, Masao Fuketa, Kazuhiro Morita, and Jun'ichi Aoe.
\newblock {A fast and compact elimination method of empty elements from a
  double-array structure}.
\newblock {\em Software: Practice and Experience}, 33(13):1229--1249, 2003.

\bibitem{pibiri2017efficient}
Giulio~Ermanno Pibiri and Rossano Venturini.
\newblock {Efficient data structures for massive n-gram datasets}.
\newblock In {\em Proceedings of the 40th International ACM SIGIR Conference on
  Research and Development in Information Retrieval}, pages 615--624, 2017.

\bibitem{pibiri2019handling}
Giulio~Ermanno Pibiri and Rossano Venturini.
\newblock {Handling massive N-gram datasets efficiently}.
\newblock {\em ACM Transactions on Information Systems}, 37(2), feb 2019.

\bibitem{Vonng-ac}
Feng Ruohang.
\newblock {ac: Aho-Corasick Automaton with Double Array Trie}.
\newblock \url{https://github.com/Vonng/ac}, 2019.
\newblock The latest version at May 24, 2022.

\bibitem{sassano2002empirical}
Manabu Sassano.
\newblock An empirical study of active learning with support vector machines
  for {Japanese} word segmentation.
\newblock In {\em Proceedings of the 40th Annual Meeting of the Association for
  Computational Linguistics (ACL)}, pages 505--512, 2002.

\bibitem{shinnou2000deterministic}
Hiroyuki Shinnou.
\newblock Deterministic japanese word segmentation by decision list method.
\newblock In {\em Proceedings of the 6th Pacific Rim International Conference
  on Artificial Intelligence (PRICAI)}, pages 822--822, 2000.

\bibitem{song2021fast}
Xinying Song, Alex Salcianu, Yang Song, Dave Dopson, and Denny Zhou.
\newblock Fast wordpiece tokenization.
\newblock In {\em Proceedings of the 2021 Conference on Empirical Methods in
  Natural Language Processing (EMNLP)}, pages 2089--2103, 2021.

\bibitem{yasuhara2013efficient}
Makoto Yasuhara, Toru Tanaka, Jun ya~Norimatsu, and Mikio Yamamoto.
\newblock An efficient language model using double-array structures.
\newblock In {\em Proceedings of the Conference on Empirical Methods in Natural
  Language Processing (EMNLP)}, pages 222--232, 2013.

\bibitem{yata2010nwc}
Susumu Yata.
\newblock {Nihongo Web Corpus 2010 (NWC 2010)}.
\newblock \url{http://www.s-yata.jp/corpus/nwc2010/}, 2010.

\bibitem{dartsclone}
Susumu Yata.
\newblock {A clone of Darts}.
\newblock \url{https://github.com/s-yata/darts-clone}, 2018.
\newblock The latest version at May 24, 2022.

\bibitem{yata2008fast}
Susumu Yata, Kazuhiro Morita, Masao Fuketa, and Jun'ichi Aoe.
\newblock {Fast string matching with space-efficient word graphs}.
\newblock In {\em Proceedings of the 4th International Conference on
  Innovations in Information Technology (IIT)}, pages 79--83, 2008.

\bibitem{yata2007efficient}
Susumu Yata, Masaki Oono, Kazuhiro Morita, Masao Fuketa, and Jun'ichi Aoe.
\newblock {An efficient deletion method for a minimal prefix double array}.
\newblock {\em Software: Practice and Experience}, 37(5):523--534, 2007.

\bibitem{yata2007compact}
Susumu Yata, Masaki Oono, Kazuhiro Morita, Masao Fuketa, Toru Sumitomo, and
  Jun'ichi Aoe.
\newblock {A compact static double-array keeping character codes}.
\newblock {\em Information Processing {\&} Management}, 43(1):237--247, 2007.

\bibitem{yata2006double}
Susumu Yata, Masaki Oono, Kazuhiro Morita, Toru Sumitomo, and Jun'ichi Aoe.
\newblock {Double-array compression by pruning twin leaves and unifying common
  suffixes}.
\newblock In {\em Proceedings of the 1st International Conference on Computing
  {\&} Informatics (ICOCI)}, pages 1--4, 2006.

\bibitem{yihan-2021-meaningfulness}
Zhou Yihan.
\newblock Meaningfulness and unit of {Z}ipf{'}s law: evidence from danmu
  comments.
\newblock In {\em Proceedings of the 20th Chinese National Conference on
  Computational Linguistics}, pages 1046--1057, August 2021.

\bibitem{yoshinaga2014self}
Naoki Yoshinaga and Masaru Kitsuregawa.
\newblock A self-adaptive classifier for efficient text-stream processing.
\newblock In {\em Proceedings of the 24th International Conference on
  Computational Linguistics (COLING)}, pages 1091--1102, 2014.

\end{thebibliography}

\end{document}